\documentclass[11pt, letterpaper]{article}
\usepackage[margin=1in]{geometry}

\usepackage{amsfonts}
\usepackage{amsmath}
\usepackage{amssymb}
\usepackage{amsthm}
\usepackage{bbm} 
\usepackage{cite}
\usepackage{color}
\usepackage{graphicx}
\usepackage{mathrsfs} 
\usepackage[normalem]{ulem}
\usepackage{wrapfig}

\def\extraspacing{\vspace{3mm} \noindent}

\def\figcapdown{\vspace{-0mm}}

\def\vgap{\vspace{1mm}}

\theoremstyle{plain}
\newtheorem{theorem}{Theorem}[section]

\newtheorem{lemma}[theorem]{Lemma}

\theoremstyle{definition}

\theoremstyle{remark}

\newcommand{\minipg}[2]{\begin{center}\begin{minipage}{#1}#2\end{minipage}\end{center}}
\newcommand{\myitems}[1]{\begin{itemize} #1 \end{itemize}}
\newcommand{\myenums}[1]{\begin{enumerate} #1 \end{enumerate}}

\newcommand{\bm}[1]{\textrm{\boldmath${#1}$}}

\newcommand{\myeqn}[1]{\begin{eqnarray}#1\end{eqnarray}}

\newcommand{\set}[1]{\{#1\}}

\def\mit{\mathit}

\def\eps{\epsilon}
\def\fr{\frac}
\def\-{\mbox{-}}

\def\real{\mathbb{R}}

\def\tO{\tilde{O}}
\def\tOmega{\tilde{\Omega}}

\def\lc{\lceil}

\def\rc{\rceil}

\def\nn{\nonumber}

\usepackage{balance}
\usepackage{cite}
\usepackage{mathtools}
\usepackage{microtype}

\def\vgap{\vspace{2mm}}
\def\figcapdown{\vspace{0mm}}
\def\minipgwidth{0.95\linewidth}

\def\A{\mathcal{A}}
\def\B{\mathcal{B}}
\def\C{\mathcal{C}}
\def\E{\mathcal{E}}
\def\G{\mathcal{G}}
\def\I{\bm{\mathcal{I}}}
\def\R{\mathcal{R}}
\def\U{\mathcal{U}}
\def\X{\mathcal{X}}
\def\agm{\mathrm{AGM}}
\def\dom{\mathit{dom}}

\def\join{\mathit{join}}
\def\no{\mathit{no}}
\def\out{\mathrm{OUT}}
\def\path{\rightsquigarrow}
\def\raw{\mathit{raw}}
\def\scheme{\mathit{scheme}}
\def\size{\mathit{size}}
\def\tri{\mathit{triangle}}
\def\wedge{\mathit{wedge}}
\def\weight{\mathit{weight}}
\def\yes{\mathit{yes}}
\title{Space-Query Tradeoffs in Range Subgraph Counting and Listing} 

\author{
	Shiyuan Deng, Shangqi Lu, Yufei Tao \\[3mm]
	Department of Computer Science and Engineering \\
	Chinese University of Hong Kong \\
	Hong Kong, China \\
	{\em \{sydeng,sqlu,taoyf\}@cse.cuhk.edu.hk}
}

\begin{document}

\maketitle

\begin{abstract}
    This paper initializes the study of {\em range subgraph counting} and {\em range subgraph listing}, both of which are motivated by the significant demands in practice to perform graph analytics on subgraphs pertinent to only selected, as opposed to all, vertices. In the first problem, there is an undirected graph $G$ where each vertex carries a real-valued attribute. Given an interval $q$ and a pattern $Q$, a query counts the number of occurrences of $Q$ in the subgraph of $G$ induced by the vertices whose attributes fall in $q$. The second problem has the same setup except that a query needs to enumerate (rather than count) those occurrences with a small delay. In both problems, our goal is to understand the tradeoff between {\em space usage} and {\em query cost}, or more specifically: (i) given a target on query efficiency, how much pre-computed information about $G$ must we store? (ii) Or conversely, given a budget on space usage, what is the best query time we can hope for? We establish a suite of upper- and lower-bound results on such tradeoffs for various query patterns. 
\end{abstract}

\vspace{50mm}
This research was supported in part by GRF Projects 14207820, 14203421, and 14222822 from HKRGC.

\pagebreak

\section{Introduction} \label{sec:intro} 

Consider $G = (V, E)$ as a {\em data graph} and $Q$ as a {\em pattern graph}. A subgraph of $G$, if isomorphic to $Q$, is said to be an {\em occurrence} of $Q$. The goal of {\em pattern searching} is to either list the occurrences of $Q$ or to count the number of them. Both are fundamental problems in computer science and have attracted considerable attention in the past few decades.     

\vgap

This paper studies pattern searching in vertex-induced subgraphs. Here, a query selects a subset $U \subseteq V$ of vertices and needs to count/list the occurrences of $Q$ in $G'$, where $G'$ is the subgraph of $G$ induced by $U$. Note that if an occurrence uses any vertex outside $U$, the occurrence should not be counted/listed. Trivially, one can answer the query by first generating $G'$ and then counting/listing $Q$ in $G'$ ``from scratch'', but this does not leverage the power of {\em preprocessing}. Instead, our goal is to store $G$ in a data structure that can answer all queries with non-trivial guarantees. It is intriguing to investigate how much we can minimize the query time subject to a space budget, and conversely, how much space we must consume to achieve a target query time. 

\vgap 

Vertex selection in database systems is done with a predicate $q$, which determines $U$ as $\set{v \in V \mid \textrm{$v$ satisfies $q$}}$. Concentrating on {\em range predicates}, the problems we consider are:

\minipg{\minipgwidth}{
	{\bf Problem 1 (Range Subgraph Counting).}
	$G = (V, E)$ is an undirected graph where each vertex $v \in V$ carries a real-valued {\em attribute} $A_v$. For an interval $q = [x_1, x_2]$, define $V_q = \set{v \in V \mid x_1 \le A_v \le x_2}$ and $G_q$ as the subgraph of $G$ induced by $V_q$. Let $Q$ be a connected (only one connected component) pattern graph with $O(1)$ vertices. Given an interval $q$, a query returns the number of occurrences of $Q$ in $G_q$. The pattern $Q$ is fixed for all queries.
} 

\minipg{\minipgwidth}{	
	{\bf Problem 2 (Range Subgraph Listing).}
	Same setup except that a query reports the occurrences of $Q$ in $G_q$. 
}

\extraspacing {\bf Universal Notations.} Several notations will apply throughout the paper. Set $n = |V|$ and $m = |E|$. Symbol $\omega < 2.37286$ \cite{aw21} represents the matrix multiplication exponent. The notations $\tO(.)$ and $\tOmega(.)$ hide a factor polylogarithmic to the underlying problem's parameters.

\subsection{Motivation} \label{sec:intro:motivation}

{\bf Practical Applications.} Subgraph patterns are important for understanding the characteristics of a data graph $G$, as has been documented in a long string of papers, e.g., \cite{ayz97,bkk17,bps21,cn85,cdm17,fkll15,flr+12,gkw19,kkm00,ww13,ayz95,bpwz14,e94,e99,els10,hkss13,knrr15,np85}. In practice, analysts are interested in not only patterns from the whole $G$ but also those pertinent only to selected vertices. Consider a social network $G$ where each vertex represents an individual. A graph's {\em clustering coefficient} \cite{ws98}, a popular measurement in network science, is the ratio between the number of triangles (3-cliques\footnote{An {\em $\ell$-clique} is a clique with $\ell$ vertices.}) and the number of wedges (2-paths\footnote{An {\em $\ell$-path} is a path with $\ell$ edges.}). The coefficient of $G$, however, is just a single value revealing little about the features of specific demographic groups. It is more informative to, for example, compare the coefficients of (i) the subgraph of $G$ induced by people with ages $\in [20, 30]$, and (ii) that induced by age $\in [60, 70]$. A step further, by putting together the coefficients induced by ``age $\in [i \cdot 10, (i+1) \cdot 10]$'' for each $i \in [1, 10]$, one obtains an interesting comparison across different age groups. Refined analysis can then concentrate on the pattern occurrences of a target group. The power of the above analysis owes to queries of Problem 1 and 2 with {\em arbitrary} selection ranges. Designing effective data structures is essential to avoid lengthy response time.

\extraspacing {\bf Importance of Space-Query Tradeoffs.} One should not confuse the space-query tradeoff with the tradeoff between {\em preprocessing time} and {\em query cost}, as has been extensively studied on join algorithms \cite{bdg07,ck21,d20,dg07,dk18,dss14,ks13,oz15,s15,ssv18,sv17}. Both tradeoffs are important, but they matter in different ways. Unlike preprocessing time, which is ``one-time cost'' (because a structure, once built, can be used forever), the space consumption is permanent. In other words, the space-query tradeoff has a (much) more durable effect on the underlying database system. However, in spite of their importance, the space-query tradeoffs on joins have received surprisingly little attention: we are aware of only a single paper \cite{dk18}, which, as will be discussed in Section~\ref{sec:related}, does not consider query predicates (or equivalently, only one query, which always outputs the entire join, exists) and concerns only reporting (but not counting). Our work can be thought of as a step in the same direction as \cite{dk18} because, as explained in Section~\ref{sec:prob2-arbitrary}, subgraph searching can be cast as a join problem (in fact, some of our results are explicitly about joins), and actually the first step on predicate-driven queries and counting.

\vgap

Finally, it is worth mentioning that a useful structure, no matter how little space it occupies, must be constructible in polynomial time. This is true for all the structures developed in our paper. In fact, each of our structures can be built with at most the time needed to find all the occurrences of the query pattern $Q$, ignoring polylog factors. 

\begin{table*}
    \begin{center}
    \begin{small}
    \begin{tabular}{c|c|c|c|c} 
        {\bf Problem} & {\bf Pattern} $Q$ & {\bf Space} & {\bf Query} & {\bf Remark} \\ 
        \hline\hline 
        1 (cnt) & any fixed $Q$ & $O(n^2)$ & $\tO(1)$ & near optimal$^\dag$ \\
        \hline 
        1 & wedge & $\tO(m^2/\lambda^2)$ & $\tO(\lambda)$ & for any $\lambda \in [1, \sqrt{m}]$, near optimal$^\dag$ \\ 
        \hline 
        1 (lower & wedge & $\tO(m^{2-\delta}/\lambda^2)$ & $\tO(\lambda)$ & for $\lambda \in [1, \sqrt{m}]$ and any $\delta > 0$, \\[-1mm] 
        bound) && impossible && subj.\ to strong set disjointness conj.\ \\
        \hline
        1 & $\ell$-clique & $O(m)$ & $\tO(1)$ & \\
        \hline\hline 
        2 (rep) & any fixed $Q$ & $\tO(m+m^{\rho^*}/\Delta)$ & delay $\tO(\Delta)$ & for any $\Delta \ge 1$,  \\[0mm]
        &&&& ${\rho^*} =$ frac.\ edge covering num.\ of $Q$ \\
        \hline
        2 & triangle & $O(m)$ & delay $\tO(1 + (m^*)^{\fr{\omega-1}{\omega+1}})$ & $m^*=$ num.\ of edges in \\[0mm]
        &&& & at least one triangle in $G_q$ \\
        \hline
        2 & $\ell$-star & $O(m)$ & delay $\tO(1)$ & near optimal \\ 
        \hline
        2 & $2\ell$-cycle & $\tO(\#P_\ell)$ & delay $\tO(1)$ & $\#P_\ell =$ num.\ of $\ell$-paths in $G$ 
    \end{tabular}
    \end{small}
    \end{center}
    \vspace{-3mm}
    {\small Remark: ``near optimal'' means no polynomial improvement (i.e., $n^\delta$ for arbitrary small constant $\delta > 0$) possbile. The near optimality marked with $\dag$ is subject to the strong set disjointness conjecture.} 
    \caption{A summary of our results} 
    \label{tab:results}
    \figcapdown
\end{table*}

\subsection{Our Contributions}  

Table~\ref{tab:results} summarizes the main results of this paper. Next, we will explain the results in detail. 

\subsubsection{Problem 1} \label{sec:ours1}

\noindent {\bf Wedges.} We will show: 

\begin{theorem} \label{thm:prob1:wedge}
	Consider Problem 1 with $Q =$ wedge. For any real value $\lambda \in [1, \sqrt{m}]$, there is a structure of $\tO(m^2/\lambda^2)$ space that answers a query in $\tO(\lambda)$ time. 
\end{theorem}

The space-query tradeoff may look disappointing. After all, wedge counting is easy in {\em one-off} computation: we can count the number of wedges in $G$ using $O(n+m)$ time. It is natural to wonder whether the space in Theorem~\ref{thm:prob1:wedge} is necessary. We answer the question by showing that any substantial improvement to Theorem~\ref{thm:prob1:wedge} will yield a major breakthrough on {\em set disjointness}:

\minipg{\minipgwidth}{
	{\bf Set Disjointness.}
	The data is a collection of $s \ge 2$ sets $S_1$, $S_2$, $...$, $S_s$. Given distinct set ids $a, b \in [1, s]$, a query returns whether $S_a \cap S_b$ is empty. 
}

\noindent Let $N = \sum_{i=1}^s |S_i|$ be the input size of set disjointness. Given any $\lambda \in [1, \sqrt{N}]$, there is a simple structure of $O(N^2 / \lambda^2)$ space with $O(\lambda)$ query time (see Appendix~\ref{app:prob1-wedge:color}). Improving the tradeoff by a polynomial factor even for one arbitrary $\lambda$ has been a long-standing open problem. The {\em strong set disjointness conjecture} \cite{gklp17,glp19} states that a structure with query time $\lambda$ must use $\tOmega(N^2 / \lambda^2)$ space for any $\lambda \ge 1$. We will prove:

\begin{theorem} \label{thm:prob1:wedge-lb} 
	Consider Problem 1 with $Q=$ wedge. Fix any $\lambda \in [1, \sqrt{m}]$ and any constant $\delta > 0$. Suppose that we can obtain a structure of $\tO(m^{2-\delta}/\lambda^2)$ space answering a query in $\tO(\lambda)$ time. Then, for any set disjointness input of size $N$, we can obtain a structure of $\tO(N^{2-\delta}/\lambda^2)$ space answering a query in $\tO(\lambda)$ time (thus breaking the strong set disjointness conjecture).
\end{theorem}

\noindent {\bf Cliques.} We will show:  

\begin{theorem} \label{thm:prob1:clique}
	For Problem 1 with $Q =$ $\ell$-clique, there is a structure of $O(m)$ space answering a query in $\tO(1)$ time. 
\end{theorem}


Counting triangles ($\ell = 3$) appears harder than counting wedges: in one-off computation, the fastest known algorithm for the former takes $O(m^{\fr{2 \omega}{\omega + 1}})$ time. It is thus surprising to see $Q=$ triangle easier than $Q =$ wedge in Problem 1.
From Theorem~\ref{thm:prob1:wedge} and \ref{thm:prob1:clique}, one sees that the problem of calculating the clustering coefficient (see Section~\ref{sec:intro:motivation}) of $G_q$ for any $q$ boils down to counting the wedges in $G_q$. Effectively, this implies optimal settlement of that problem (subject to the strong set disjointness conjecture), which bears practical significance due to the popularity of clustering coefficients. 

\extraspacing {\bf Arbitrary Subgraphs.} We will show:

\begin{theorem} \label{thm:prob1:gen}
	For any $Q$, there is a structure for Problem 1 that uses $O(n^2)$ space and answers a query in $\tO(1)$ time. 
\end{theorem}

The above result is difficult to improve: reducing the space by an $n^\delta$ factor for any constant $\delta > 0$ breaks the strong set disjointness conjecture. To explain, assume $n = O(m)$.\footnote{Discard ``isolated'' vertices with no incident edges.} If there was a structure of $O(n^{2-\delta}) = O(m^{2-\delta})$ space and $\tO(1)$ query time, applying the structure to $Q=$ wedge would yield a breakthrough on set disjointness by way of  Theorem~\ref{thm:prob1:wedge-lb}. The reader should note that the hardness comes from producing a guarantee on all $Q$; it is possible to do better for special patterns (Theorem~\ref{thm:prob1:clique}). The hardness thus endows $Q =$ wedge with unique significance in Problem 1. Theorem~\ref{thm:prob1:gen} further implies that Problem 1 under $Q =$ wedge is the hardest when $G$ is the sparsest: $m = o(n^{1+\eps})$ for any constant $\eps > 0$. To see why, set $m = n^{1+\eps}$, which gives $n^2 = m^{\fr{2}{1+\eps}}$. Since $\fr{2}{1+\eps} = 2 - \fr{2\eps}{1+\eps}$, Theorem~\ref{thm:prob1:gen} yields a structure of $O(m^{2-\delta})$ space and $\tO(1)$ query time with $\delta = \fr{2\eps}{1+\eps}$, improving Theorem~\ref{thm:prob1:wedge} by a polynomial factor at $\lambda = \tO(1)$.    

\subsubsection{Problem 2} \label{sec:ours2} 

A listing query ensures a {\em delay} $\Delta$ if it reports a new occurrence of $Q$ or declares ``no more occurrences'' within $\Delta$ time after the previous occurrence\footnote{The reader may assume that a dummy occurrence is always output at the beginning of a query algorithm.}. 

\extraspacing {\bf Arbitrary Subgraphs.} We will show:

\begin{theorem} \label{thm:prob2:gen}
	For any $Q$ and $\Delta \ge 1$, there is a structure for Problem 2 that uses $\tO(m + m^{\rho^*}/\Delta)$ space and has a query delay of $\tO(\Delta)$, where ${\rho^*}$ is the fractional edge covering number of $Q$.
\end{theorem}

Imagine assigning each edge of $Q$ a non-negative weight such that (i) for each vertex of $Q$, all its incident edges receive a combined weight at least 1 and (ii) the total weight of all edges is minimized. The fractional edge covering number ${\rho^*}$ of $Q$ is the total weight of an optimal assignment. The maximum number of occurrences of $Q$ in $G$ is $O(m^{\rho^*})$ \cite{agm13} and the bound is tight in the worst case. 

\vgap

Our structure actually settles a problem on natural joins:

\minipg{\minipgwidth}{
	{\bf Range Join.} Let $\R$ be a set of $O(1)$ relations each with $O(1)$ real-valued attributes. Denote by $\join(\R)$ the natural join result on the relations in $\R$.    
	Given an interval $q = [x_1, x_2]$, a query reports all the tuples $t \in \join(\R)$ such that every attribute of $t$ falls in $q$.
	
}

\noindent Let $N$ be the total number of tuples in the relations of $\R$. For any $\Delta \ge 1$, we give a structure of $\tO(N + N^{\rho^*}/\Delta)$ space answering a query with an $\tO(\Delta)$ delay. Here, the fractional edge covering number ${\rho^*}$ is with respect to the join's hypergraph (details deferred to Section~\ref{sec:prob2-arbitrary}). 

\vgap 

The challenge behind Theorem~\ref{thm:prob2:gen} is to design a structure that works for all $Q$. It is possible to do better for specific $Q$. Next, we present three examples that are not only important subproblems themselves but also illustrate different techniques.

\extraspacing {\bf Triangles.} We will show:

\begin{theorem} \label{thm:prob2:triangle}
	For Problem 2 with $Q = $ triangle, there is a structure of $O(m)$ space answering a query with an $\tO(1 + (m^*)^{\fr{\omega-1}{\omega+1}})$ delay, where $m^*$ is the number of edges appearing in at least one reported triangle. 
\end{theorem}

The fractional edge covering number ${\rho^*}$ is 1.5 for $Q = $ triangle. To ensure $\tO(m)$ space, Theorem~\ref{thm:prob2:gen} needs to set $\Delta = \sqrt{m}$. As $\fr{\omega-1}{\omega+1} < 0.408$, Theorem~\ref{thm:prob2:triangle} achieves a polynomial improvement in delay. The reader should note that the value $m^*$ in Theorem~\ref{thm:prob2:triangle} never exceeds $m$ but can be much less (this happens when there are few triangles to list). Problem 2 with $Q =$ triangle and $q$ fixed to $(-\infty, \infty)$ was used as a motivating problem in the previous work of \cite{dk18}, which described a structure of $O(m)$ space with a delay $\tO(\sqrt{m})$ and is thus strictly improved by Theorem~\ref{thm:prob2:triangle}. 

\extraspacing {\bf $\bm{\ell}$-Stars.} An {\em $\ell$-star} is a tree with $\ell$ leaves and one non-leaf vertex (a wedge is a 2-star). We will show:

\begin{theorem} \label{thm:prob2:l-star}
	For Problem 2 where $Q = \ell$-star, there is a structure of $O(m)$ space answering a query with an $\tO(1)$ delay. 
\end{theorem}

As a corollary, for any interval $q$, $O(m)$ space suffices to {\em detect} the presence of an $\ell$-star in $G_q$ using $\tO(1)$ time. For $Q=$ wedge, this means that the hardness manifested by Theorem  \ref{thm:prob1:wedge-lb} is indeed due to {\em counting}.

\extraspacing {\bf $\bm{2\ell}$-Cycles\footnote{A cycle with $2\ell$ vertices.}.} We will show:

\begin{theorem} \label{thm:prob2:cycle}
	For Problem 2 with $Q = 2\ell$-cycle where $\ell \ge 2$, there is a structure of $\tO(\# P_\ell)$ space answering a query with an $\tO(1)$ delay, where $\# P_\ell$ is the number of $\ell$-paths in $G$.
\end{theorem}

The fractional edge covering number ${\rho^*}$ is $\ell$ for a $2\ell$-cycle. Theorem~\ref{thm:prob2:gen} needs $\tO(m^\ell)$ space to achieve an $\tO(1)$ delay. The space in Theorem~\ref{thm:prob2:cycle} is significantly better. For $\ell = 2$ ($Q=$ 4-cycle), the space is $\tO(nm)$ which is the maximum number of wedges in $G$. For $\ell > 2$, the space is $\tO(m^{\lc(\ell+1)/2\rc})$ which is the maximum number of $\ell$-paths in $G$. 

%

\subsection{Related Work} \label{sec:related}

The preceding sections have covered the most relevant existing results. We will now proceed to discuss other related work. 

\vgap

Pattern searching has been extensively studied in one-off computation. We refer the reader to \cite{ayz97,bkk17,bps21,cn85,cdm17,fkll15,flr+12,gkw19,kkm00,ww13} and \cite{ayz95,bpwz14, cn85,e94,e99,els10,hkss13,knrr15,np85,nprr18}, as well as the references therein, for algorithms on counting and listing, respectively. Those algorithms can be applied in Problem 1 and 2 after $G_q$ has been generated. Our focus in this work is to avoid a full generation of $G_q$ because doing so can take $\Omega(m)$ time. 

\vgap 

In the other extreme, one can precompute the set $S$ of occurrences of $Q$ in $G$. The size of $S$ is $O(m^{\rho^*})$ (AGM bound), assuming that $Q$ has a constant size. By resorting to standard computational geometry techniques \cite{bcko08}, we can store $S$ in structures of $\tO(m^{\rho^*})$ space to answer a query of Problem 1 in $\tO(1)$ time and a query of Problem 2 with an $\tO(1)$ delay. For Problem 1, Theorem~\ref{thm:prob1:gen} achieves a better space bound on every $Q$ with ${\rho^*} \ge 2$. When ${\rho^*} < 2$, $Q$ has at most three vertices: a 1-path (single edge), a wedge, or a triangle. We have resolved the wedge and triangle cases (Theorem \ref{thm:prob1:wedge} and \ref{thm:prob1:clique}), while Problem 1 is trivial for $Q=$ 1-path. For Problem 2, Theorem~\ref{thm:prob2:gen} captures the above extreme idea as a special case with $\Delta = \tO(1)$ and offers a tunable space-query tradeoff. 

\vgap 

A {\em relational event graph}, introduced by Bannister et al.\ \cite{bdes13}, is a graph $G = (V, E)$ where every {\em edge} $e \in E$ carries a real-valued timestamp. For an interval $q = [x_1, x_2]$, let $G^\mit{edge}_q$ be the subgraph of $G$ induced by all the edges whose timestamps are covered by $q$. A pattern searching query counts/lists the occurrences of a pattern $Q$ in $G^\mit{edge}_q$. See \cite{bdes13, cm19, cms20} for several data structures designed for such queries. Similar as it sounds, pattern searching on a relational event graph is drastically different from Problem 1 and 2 such that there is little overlap --- in neither results nor techniques --- between our solutions and those in \cite{bdes13, cm19, cms20}.

\vgap

Delay minimization is an important topic in the literature of joins and conjunctive queries; see \cite{bdg07,bks17,ck21,czb+20,d20,dss14,dg07,dk18,kno+20,ks13,oz15,s15b,ssv18,sv17} and their references. Regarding our problems, we are not aware of previous work giving a result better than what has already been mentioned. Our formulation of range join listing (Section~\ref{sec:ours2}) suggests that the presence of query predicates can pose new challenges on joins (also conjunctive queries) from the indexing's perspective. Deep and Koutris \cite{dk18} proved a result equivalent to Theorem~\ref{thm:prob2:gen} (up to an $\tO(1)$ factor) on Problem 2, but only in the special scenario where a query concerns the whole $G$, i.e., fixing the query range $q$ to $(-\infty, \infty)$. 

\section{Preliminaries} \label{sec:pre}

In this section, we will describe several technical tools to be deployed in our solutions. 

\extraspacing {\bf Structures for Multidimensional Points.} We will utilize some well-known geometry data structures as introduced below. The reader does not need to be bothered with the details of these structures because we will apply them as ``black boxes''. Let $P$ be a set of $n$ points in $d$-dimensional space $\real^d$ where $d$ is a constant. Given a rectangle $q$ of the form $[x_1, y_1] \times [x_2, y_2] \times ... \times [x_d, y_d]$, a {\em range reporting} query enumerates the points in $P \cap q$. We can create a {\em range tree} \cite{b79,bcko08} on $P$, which uses $\tO(n)$ space and permits us to answer such a query with an $\tO(1)$ delay. When $d = 2$, we can replace the range tree with a {\em Chazelle's structure} \cite{c88} which retains the aforementioned query performance but reduces the space consumption to $O(n)$.

\vgap

We will also need {\em range sum} queries on $P$ in the scenario where each point in $P$ is 2D (i.e., $d=2$) and carries a real-valued {\em weight}. Given a rectangle $q = [x_1, y_1] \times [x_2, y_2]$, such a query reports the total weight of the points in $P \cap q$. We can again build a Chazelle's structure of \cite{c88} on $P$ which occupies $O(n)$ space and answers a query in $\tO(1)$ time. 

\extraspacing {\bf From ``Delays with Duplicates'' to ``Delays under Distinctness''.} Let us consider a duplicate-removal scenario often encountered in designing algorithms with small delays. Suppose that we have an algorithm $\A$ for enumerating a set $S$ of elements. With a delay of $\Delta$, $\A$ can report an element $e \in S$, but cannot guarantee that $e$ has never been reported before. The good news, on the other hand, is that $\A$ can output the same element at most $\alpha$ times for some $\alpha \ge 1$ .

\vgap 

By modifying a {\em buffering technique} in \cite{ty22}, we can convert $\A$ into an  algorithm that enumerates only the {\em distinct} elements of $S$ with a delay of $O(\alpha \cdot \Delta \log|S|)$. Conceptually, divide the execution of $\A$ into {\em epochs}, each of which runs for $\alpha \cdot \Delta$ time\footnote{Recall that ``time'' in the RAM model is defined as the number of atomic operations (e.g., addition, multiplication, comparison, accessing a memory word, etc.) executed. Each epoch is essentially a sequence of $\alpha \cdot \Delta$ such operations.}. As $\A$ runs, we use a {\em buffer} $B$ to stash the set of distinct elements that have been found by $\A$ but not yet reported. Every time $\A$ finds an element $e \in S$, we check whether $e$ has ever existed in $B$ (this takes $O(\log|S|)$ time, using a binary search tree maintained on all the elements that have ever been found so far). If so, $e$ is ignored; otherwise, it is added to $B$. At the end of each epoch, we output an arbitrary element from $B$ and remove it from $B$. Finally, after $\A$ has terminated, we simply output the remaining elements in $B$. 

\vgap 

$B$ always contains at least one element at the end of each epoch. To see why, consider the end of the $t$-th epoch for some $t \ge 1$. At this moment, $\A$ has been running for $t \cdot \alpha \cdot \Delta$ time and therefore must have reported $t \cdot \alpha$ elements, which may not be distinct. However, as each element can be reported at most $\alpha$ times, there must be at least $t$ (distinct) ones among those $t \cdot \alpha$ elements. Since we have reported only $t - 1$ elements in the preceding epochs, $B$ must still have at least one element at the end of epoch $t$. It is now straightforward to verify that the modified algorithm has a delay of $O(\alpha \cdot \Delta \log|S|)$ in enumerating the distinct elements of $S$.

\section{Problem 1: Matching Upper and Lower Bounds} \label{sec:prob1-lb-ub} 

This section will establish the conditional lower bound in Theorem~\ref{thm:prob1:wedge-lb} and its matching upper bound in Theorem~\ref{thm:prob1:gen}. Our discussion on the upper bound will also establish Theorem~\ref{thm:prob1:clique}. Throughout the paper, we will assume that the vertices of $G$ have distinct attribute values. The assumption loses no generality because one can break ties by vertex id.

\subsection{Lower Bound} \label{sec:prob1-lb-ub:lb}

Suppose that Problem 1 under $Q =$ wedge admits a structure that uses $\tO(m^{2-\delta}/\lambda^2)$ space and answers a query in $\tO(\lambda)$ time for some $\lambda \ge 1$. We will design a structure for set disjointness that uses $\tO(N^{2-\delta}/\lambda^2)$ space and answers a query in $\tO(\lambda)$ time. Recall that the data input to set disjointness consists of $s \ge 2$ sets $S_1, ..., S_s$ with a total size of $N$. Define $\U = \bigcup_{i=1}^s S_i$. 

\vgap

Create a graph $G = (V, E)$ as follows. $V$ has $2s + |\U|$ vertices, including $2s$ {\em set vertices} and $|\U|$ {\em element vertices}. Each set $S_i$ ($i \in [1, s]$) defines two set vertices, whose attribute values are set to $i$ and $s + i$, respectively. Each element in $\U$ defines an element vertex with the same attribute value $s + 1/2$. Set $E$ contains $2N$ edges: for each element $e \in S_i$, add to $E$ two edges each between the element vertex of $e$ and a set vertex of $S_i$. Now, create a Problem-1 structure under $Q =$ wedge on $G$. The structure occupies $\tO(N^{2-\delta}/\lambda^2)$ space. 

\vgap 

Consider a set disjointness query with set ids $a$ and $b$. Assuming w.l.o.g.\ $a < b$, we issue four Problem-1 wedge-counting queries on $G$ with intervals $q_1 = [a, s+b]$, $q_2 = [a+1, s+b]$, $q_3 = [a, s+b-1]$, and $q_4 = [a+1, s+b-1]$, respectively. Let $c_1, c_2, ..., c_4$ be the counts returned. We declare $S_a \cap S_b$ non-empty if and only if $c_1 - c_2 - c_3 + c_4 > 0$. The query time is $\tO(\lambda)$. Appendix~\ref{app:prob1-lb:correctness} proves the algorithm's correctness. This completes the proof of Theorem~\ref{thm:prob1:wedge-lb}.

\subsection{Upper Bound} \label{sec:prob1-lb-ub:ub} 

Next, we will attack Problem 1 by allowing $Q$ to be an arbitrary pattern graph. Consider any occurrence of $Q$ in $G$. Let $u$ (resp.\ $v$) be the vertex in this occurrence with the smallest (resp.\ largest) attribute. We {\em register} the occurrence at the pair $(u, v)$. Denote by $c_{u,v}$ the number of occurrences registered at $(u,v)$. 

\vgap 

For a query with $q = [x_1, x_2]$, an occurrence registered at $(u, v)$ appears in $G_q$ (i.e., the subgraph of $G$ induced by $V_q$) if and only if $A_u \ge x_1$ and $A_v \le x_2$. We can therefore convert the problem to {\em range sum} on 2D points. For each pair $(u, v) \in V \times V$, create a point $(A_u, A_v)$ with weight $c_{u,v}$. Let $P$ be the set of points created; clearly, $|P| = O(n^2)$. The query result is simply the total weight of all the points in $P$ covered by the rectangle $[x_1, \infty) \times (-\infty, x_2]$ (a range sum operation). We can store $P$ in a Chazelle's structure (see Section~\ref{sec:pre}) that occupies $O(|P|) = O(n^2)$ space and performs a range sum operation in $\tO(1)$ time. This establishes Theorem~\ref{thm:prob1:gen}. 

\extraspacing {\bf Improvement for Cliques.} The space of our structure can be lowered to $O(m)$ when $Q$ is a clique. The crucial observation is that registering an occurrence at $(u,v)$ implies  $\set{u,v} \in E$. We add to $P$ only the points $(A_u, A_v)$ with a non-zero $c_{u,v}$ (points with zero weights do not affect a range sum operation). This reduces the size of $P$ to at most $m$ and, hence, the space of the Chazelle's structure to $O(m)$. We thus complete the proof of Theorem~\ref{thm:prob1:clique}.

\section{Problem 1: Wedges} \label{sec:prob1-wedge}

The section will explain how to achieve the guarantees in Theorem~\ref{thm:prob1:wedge} for Problem 1 under $Q =$ wedge. We will represent a wedge occurrence in $G = (V, E)$ as $\wedge(u,v,w)$ where $u, v,$ and $w$ are vertices in $V$, and $\set{u,v}$ and $\set{v, w}$ are edges in $E$. Let us introduce a slightly different problem: 

\minipg{\minipgwidth}{
	{\bf Colored Range Wedge Counting.}
	Define $G = (V, E)$ and $A_v$ for each $v \in V$ as in Problem 1. Each vertex in $V$ is colored black or white. Given an interval $q$, a query returns the number of occurrences $\wedge(u,v,w)$ such that $A_u \in q$, $A_w \in q$, and $v$ is black. 
}

\noindent Note that no requirements exist on $A_v$ and the colors of $u$ and $w$. 

\vgap

Let $\C$ be a set of subsets of $V$. We call $\C$ a {\em canonical collection} if
\myitems{
	\item (P4-1) each vertex of $V$ appears in $\tO(1)$ subsets in $\C$; 
	\item (P4-2) for any interval $q$, we can partition $V_q$ (i.e., the set of vertices in $V$ with attribute values in $q$) into $\tO(1)$ disjoint subsets, each being a member of $\C$. The ids of these subsets can be obtained in $\tO(1)$ time. 
}
It is rudimentary to find a canonical collection $\C$ satisfying $\sum_{U \in \C} |U| = \tO(n)$.\footnote{It suffices to build a binary search tree $T$ on the vertices' attribute values. Each node in $T$ defines a subset in $\C$, which consists of every $v \in V$ whose attribute $A_v$ is stored in the node's subtree. It is well known (see, e.g., \cite{t22}) that, for any interval $q$, there exist $O(\log n)$ {\em canonical nodes} in $T$ whose subtrees are disjoint and together contain all and only the attribute values in $q$. Those nodes can be found in $O(\log n)$ time and satisfy Property P4-2 with respect to $V_q$.} We will work with such a $\C$ henceforth. In Appendix~\ref{app:prob1-wedge:color}, we prove: 

\begin{lemma} \label{lmm:prob1-wedge:color}
	Consider the colored range wedge counting problem. For any real value $\lambda \in [1, \sqrt{m}]$, there is a structure of $\tO(m^2/\lambda^2)$ space that answers a query in $\tO(\lambda)$ time. 
\end{lemma}

Equipped with the above, we now return to Problem 1 with $Q =$ wedge. 

\extraspacing {\bf Structure.} For each $U \in \C$ (where $U$ is a subset of $V$), we create a graph $G_U$ by adding edges in three steps:
\myenums{
    \item Initialize $G_U$ as an empty graph with no vertices and edges.
    \item For every vertex $u \in U$, we add all its edges in $G$ (i.e., the original data graph) to $G_U$. The addition of an edge $\set{u,v}$ creates vertex $v$ in $G_U$ if $v$ is not present in $G_U$ yet. 
    \item Finally, color a vertex in $G_U$ black if it comes from $U$, or white otherwise.
}
We now build a structure of Lemma~\ref{lmm:prob1-wedge:color} on $G_U$, which uses $\tO(|E_U|^2 / \lambda^2)$ space where $E_U$ is the set of edges in $G_U$. By Property P4-1, each edge $\set{u,v}$ of $G$ can be added to the $E_U$ of $\tO(1)$ subsets $U \in \C$. It thus follows that $\sum_{U \in \C} |E_U| = \tO(m)$. The structures of all $U \in \C$ occupy $\tO(m^2 / \lambda^2)$ space in total. 

\extraspacing {\bf Query.} Consider now a (Problem-1) query with interval $q$. By Property P4-2, in $\tO(1)$ time we can pick $h = \tO(1)$ members $U_1, ...,U_h$ from $\C$ to partition $V_q$. For each $i \in [1, h]$, issue a colored range wedge counting query with interval $q$ on $G_{U_i}$. We return the sum of the $h$ queries' outputs. The overall query time is $h \cdot \tO(\lambda) = \tO(\lambda)$.

\vgap 

To verify correctness, first observe that every $\wedge(u,v,w)$ counted by the colored query on $G_{U_i}$ satisfies: $A_u \in q$, $A_w \in q$ (definition of colored range wedge counting), and $A_v \in q$ (because $v$ being black means $v \in U_i \subseteq V_q$). Conversely, every occurrence $\wedge(u,v,w)$ satisfying $\{A_u, A_v, A_w\} \subseteq q$ is counted only once: by the colored query on $G_{U_i}$ where $U_i$ is the only subset (among all $i \in [1, h]$) containing $v$. Indeed, for any $U_j$ with $j \ne i$, $v$ is either absent in $G_{U_j}$ or is white; in neither case can the wedge be counted. Correctness now follows. 



\section{Problem 2: Arbitrary Subgraphs} \label{sec:prob2-arbitrary} 

We now proceed to tackle Problem 2 for an arbitrary query pattern $Q$. We will, in fact, solve the range join problem defined in Section~\ref{sec:ours2}. As shown in Appendix \ref{app:prob1-arbitrary}, it is relatively easy to convert our structure to prove Theorem~\ref{thm:prob2:gen}.

\vgap 

For a relation $R \in \R$ (recall that $\R$ is the set of input relations; see Section~\ref{sec:ours2}) its scheme, $\scheme(R)$, is the set of attributes in $R$. Let $\X = \bigcup_{R \in \R} \scheme(R)$. The input size $N$ can now be expressed as $\sum_{R \in \R} |R|$. We will assume, w.l.o.g., that (i) the relations in $\R$ have distinct schemes, (ii) $N$ is a power of 2, and (iii) each attribute $X \in \X$ has a domain $\dom(X)$ comprising the integers in $[1, N]$. Given an interval $q = [x_1, x_2]$, a query lists every tuple $t$ in $\join(\R)$ --- the natural join result on $\R$ --- satisfying $t[X] \in q$ for all $X \in \X$, where $t[X]$ is the tuple's value under attribute $X$. We want to design a structure of small space to answer such queries with a small delay.

\vgap

It will be convenient to work with a hypergraph $\G = (\X, \E)$ where $\E = \set{\scheme(R) \mid R \in \R}$. Given an edge $e \in \E$, we use $R_e$ to denote the (only) relation $R \in \R$ whose scheme is $e$. For a function $W$ that assigns a non-negative weight $W(e)$ to every $e \in \E$, its {\em lump-sum} is $\sum_{e \in \E} W(e)$. The function $W$ is a {\em fractional edge covering} if $\sum_{e \in \E: X \in e} W(e) \ge 1$ holds on every attribute $X \in \X$. The {\em fractional edge covering number} ${\rho^*}$ of $\G$ is the smallest lump-sum of all fractional edge coverings. Henceforth, we will use $W$ to represent an optimal assignment function with lump-sum ${\rho^*}$. 

\vgap 

The section's main result is: 

\begin{theorem} \label{thm:range-join}
	For the range join problem (see Section~\ref{sec:ours2}), given any $\Delta \ge 1$, there is a structure of $\tO(N + N^{\rho^*}/\Delta)$ space that answers a query with an $\tO(\Delta)$ delay.
\end{theorem}

\subsection{A Generalization of the AGM Bound} \label{sec:prob2-arbitrary:agm}

The classical AGM bound \cite{agm13} states that $|\join(\R)| \le \prod_{e \in \E} |R_e|^{W(e)}$. Next, we will present a more general version of this inequality. 

\vgap 

Set $d = |\X|$ and impose an arbitrary ordering on the $d$ attributes: $X_1, X_2, ..., X_d$. Given intervals $I_1, I_2, ..., I_d$ where $I_i \subseteq \dom(X_i)$ for each $i \in [1, d]$, define $B(I_1, ..., I_d)$ as the $d$-dimensional box $I_1 \times ... \times I_d$. For a relation $R \in \R$, we use $R \ltimes B(I_1, ..., I_d)$ to represent the set of tuples $t \in R$ such that $t[X_i] \in I_i$ for every $i$ satisfying $X_i \in \scheme(R)$. 

\vgap 

We prove in Appendix~\ref{app:prob2-arbitrary:agm-gen}:

\begin{lemma} \label{lmm:prob2-arbitrary:agm-gen}
	Let $\I_i$, $i \in [1, d]$, be a set of disjoint intervals in $\dom(X_i)$. Then: 
	\myeqn{
		\sum_{I_1 \in \I_1}\sum_{I_2 \in \I_2}...\sum_{I_d \in \I_d}
		\prod_{e \in \E} 
		|R_e \ltimes B(I_1, ..., I_d)|^{W(e)}
		\le 
		\prod_{e \in \E} |R_e|^{W(e)}.
		\label{eqn:agm-gen}
	}
\end{lemma}

To see how \eqref{eqn:agm-gen} captures the AGM bound, consider the special $\I_i$ with size $|\dom(X_i)|$, namely, each interval in $\I_i$ is a value in $\dom(X_i)$ and vice versa.   Thus, $|R_e \ltimes B(I_1, ..., I_d)|$ is either 0 or 1 such that the left hand side of \eqref{eqn:agm-gen} is precisely $|\join(\R)|$. The real power of \eqref{eqn:agm-gen}, however, comes from allowing $\I_i$ to be an arbitrary set of disjoint intervals, a feature crucial for us to prove Theorem~\ref{thm:range-join}.

\vgap 

A remark is in order about why Lemma~\ref{lmm:prob2-arbitrary:agm-gen} is not trivial. It would be if the term $\prod_{e \in \E} |R_e \ltimes B(I_1, ..., I_d)|^{W(e)}$ in \eqref{eqn:agm-gen} was replaced by the output size of the join on the relations in $\{R_e \ltimes B(I_1, ..., I_d) \mid e \in \E\}$. By the AGM bound, the term $\prod_{e \in \E} |R_e \ltimes B(I_1, ..., I_d)|^{W(e)}$ is an {\em upper bound} on the size of the join $\{R_e \ltimes B(I_1, ..., I_d) \mid e \in \E\}$. The non-trivial goal is to show that the summation of all those {\em upper} bounds (i.e., the left hand side of \eqref{eqn:agm-gen}) still cannot exceed $\prod_{e \in \E} |R_e|^{W(e)}$.  

\subsection{Range Join} \label{sec:prob2-arbitrary:range-join}

This subsection serves as a proof of Theorem~\ref{thm:range-join}. Given an $\ell \ge 0$, we call an interval a {\em level-$\ell$ dyadic interval} if it has the form $[i \cdot 2^\ell + 1,$ $(i+1) \cdot 2^\ell]$ for some integer $i \ge 0$. Because $N$ is a power of 2, for each $\ell \in [0, \log_2 N]$, we can partition $[1, N]$ into $N/2^\ell$ disjoint level-$\ell$ dyadic intervals. 

\vgap 

A {\em dyadic combination} is a sequence of $d$ dyadic intervals $(I_1, ..., I_d)$; recall that $d = |\X|$. The combination defines a (natural) join instance on the relations in $\{R_e \ltimes B(I_1, ..., I_d) \mid e \in \E\}$. We will denote the instance as $\R_{I_1, ..., I_d}$. Define 
\myeqn{
	\agm(I_1, ..., I_d) 
	&=&
	\prod_{e \in \E} |R_e \ltimes B(I_1, ..., I_d)|^{W(e)}. \label{eqn:range-join:agm-subjoin}
}
The AGM bound assures us that $|\join(\R_{I_1, ..., I_d})| \le \agm(I_1, ..., I_d)$. 

\extraspacing {\bf Structure.} A dyadic combination $(I_1, ..., I_d)$ with a non-empty $\join(\R_{I_1, ..., I_d})$ is said to be {\em heavy} if $\agm(I_1, ..., I_d) > \Delta$, or {\em light} otherwise. For each heavy combination, we build a structure of \cite{dk18} that can enumerate the tuples in $\join(\R_{I_1, ..., I_d})$ with an $\tO(\Delta)$ delay. The structure's space is bounded by $O(\agm(I_1, ..., I_d)/\Delta)$.\footnote{Strictly speaking, the space should also account for the  relations in $\R_{I_1,...,I_d}$. In our context, it suffices to store the relations of $\R$ once and generate the relations in $\R_{I_1,...,I_d}$ when answering a query. Appendix~\ref{app:prob1-arbitrary} has additional details about \cite{dk18}.}

\vgap 

We argue that the structures on all the heavy (dyadic) combinations use $\tO(N^{\rho^*}/\Delta)$ space in total. Fix $d$ arbitrary level numbers $\ell_1, ..., \ell_d$ each between $0$ and $\log_2 N$. For $i \in [1, d]$, let $\I_i$ be the set of all level-$\ell_i$ dyadic intervals. The total space occupied by the structures of all heavy combinations $(I_1, ..., I_d) \in \I_1 \times ... \times \I_d$ is 
\myeqn{
	\fr{1}{\Delta} \sum_{(I_1, ..., I_d) \in \I_1 \times ... \times \I_d}
	\agm(I_1, ..., I_d).
	\label{eqn:range-join:space}
}
up to an $\tO(1)$ factor. The above includes a term for every light combination but such terms can only over-estimate the space. Each $\I_i$ is a set of disjoint intervals in $\dom(X_i)$. Applying the definition in \eqref{eqn:range-join:agm-subjoin} and Lemma~\ref{lmm:prob2-arbitrary:agm-gen}, we can see that \eqref{eqn:range-join:space} is bounded by $N^{\rho^*}/\Delta$, noticing that the right hand side of \eqref{eqn:agm-gen} is at most $N^{\rho^*}$. 

\vgap 

In the above analysis, we have fixed a set of $\ell_1, ..., \ell_d$. As each $\ell_i$ has $O(\log N)$ choices, all together there are $O(\log^d N) = \tO(1)$ different sets of $\ell_1, ..., \ell_d$. We can now conclude that the overall space is $\tO(N^{\rho^*}/\Delta)$.  

\vgap 

Finally, we need a hash table to check in constant time whether a dyadic combination is heavy. The hash table occupies $\tO(N^{\rho^*}/\Delta)$ space because our earlier analysis implies a bound $\tO(N^{\rho^*}/\Delta)$ on the number of heavy dyadic combinations. The overall space of our entire structure is therefore $\tO(N + N^{\rho^*}/\Delta)$, where the term $\tO(N)$ counts the space for storing the relations of $\R$.

\extraspacing {\bf Query.} Consider a range join query with interval $q = [x_1, x_2]$. We consider, w.l.o.g., that $x_1$ and $x_2$ are integers in $[1,N]$. In $\tO(1)$ time, we can partition the box $B(\underbrace{q, ..., q}_t)$ into $O(\log^d N) = \tO(1)$ disjoint boxes, each in the form $B(I_1,...,I_d)$ where $(I_1,...,I_d)$ is a dyadic combination; we say that $(I_1,...,I_d)$ is {\em canonical} for $q$. The query result is
\myeqn{
	\bigcup_{\text{canonical $(I_1,...,I_d)$}} \join(\R_{I_1,...,I_d}). \nn
}
The results $\join(\R_{I_1,...,I_d})$ of all the canonical $(I_1,...,I_d)$ are disjoint. If a canonical $(I_1,...,I_d)$ is heavy, we enumerate $\join(\R_{I_1,...,I_d})$ with an $\tO(\Delta)$ delay using the structure of \cite{dk18} on $(I_1,...,I_d)$. Otherwise, we apply a worst-case optimal join algorithm \cite{nprr18,nrr13,v14} to compute $\join(\R_{I_1,...,I_d})$. The algorithm finishes in $\tO(\agm(I_1,...,I_d))$ time, which is $\tO(\Delta)$ by definition of light dyadic combination. Our algorithm guarantees a delay of $\tO(\Delta)$. This completes the proof of Theorem~\ref{thm:range-join}. 


\extraspacing {\bf Remark.} In \cite{knrr15}, Khamis et al.\ used dyadic intervals in their algorithm for one-off computation of $\join(\R)$. Their main technical issue was to select  ``good'' dyadic boxes (i.e., boxes of the form $B(I_1, ..., I_d)$) to cover the tuples in $\join(\R)$ once. That issue is non-existent in our context, where the primary obstacle is  to argue that the total space given in \eqref{eqn:range-join:space} is affordable. We overcame the obstacle using Lemma~\ref{lmm:prob2-arbitrary:agm-gen}, which, though perpahs no longer surprising given all the existing variations of the AGM bound, deserves a careful treatment that, we believe, has not appeared before.


\section{Problem 2: Triangles} \label{sec:prob2-triangle} 

This section will describe a structure for Problem 2 under $Q =$ triangle. We will first attack, in Section \ref{sec:prob2-triangle:range-tri-edges} and \ref{sec:prob2-triangle:sdtl}, two fundamental problems whose solutions are vital to establishing Theorem~\ref{thm:prob2:triangle}, the proof of which is presented in Section \ref{sec:prob2-triangle:tri}. 


\subsection{The Range Triangle Edges Problem} \label{sec:prob2-triangle:range-tri-edges} 

This subsection will discuss the following standalone problem.

\minipg{\minipgwidth}{
	{\bf Range Triangle Edges (RTE).} Let $G$ be an undirected graph with $m$ edges. Given an interval $q = [x_1, x_2]$, a query returns: (i) all the edges appearing in at least one triangle of $G_q$; and (ii) $\Theta(m^*)$ triangles where $m^*$ is the number of edges reported in (i). 
}

\noindent We will develop a structure of $O(m)$ space that can answer a query in $\tO(m^*)$ time. Furthermore, the query can enumerate the $m^*$ edges and the $\Theta(m^*)$ triangles both with a delay $\Delta$. 

\vgap 

Let us represent a triangle occurrence in $G$ as $\tri(u,v,w)$ where $u, v,$ and $w$ are the triangle's vertices. Ordering is important: we will always adhere to the convention $A_u < A_v < A_w$. Given an interval $q$, we denote by $E^*_q$ the set of edges showing up in at least one triangle of $G_q$. Hence, $m^* = |E^*_q|$. If $\tri(u,v,w)$ appears in $G_q$, we call $\set{u,v}$ a type-1 edge, $\set{v,w}$ a type-2 edge, and $\set{u,w}$ a type-3 edge. The total number of edges of all three types is between $m^*$ and $3m^*$.\footnote{An edge can be of different types in various triangle occurrences.}. Next, we explain how to extract the edges of each type in $G_q$.

\extraspacing {\bf Type 1 and 2.} We will discuss only type 1 because type 2 is symmetric. For each edge $\set{u,v}$ in $G$ (assume, w.l.o.g., $A_u < A_v$), identify a {\em sentinel} vertex $w^*$ for $\set{u,v}$ as follows:
\myitems{
	\item $w^* =$ null if $G$ has no occurrence of the form $\tri(u,v,w)$; 
	\item otherwise, $w^*$ has the smallest attribute among all the vertices $w$ making a triangle occurrence $\tri(u,v,w)$ in $G$.
}

Consider any interval $q = [x_1, x_2]$. Observe that $\set{u,v}$ is a type-1 edge for $q$ if and only if $x_1 \le A_u$ and $A_{w^*} \le x_2$. This motivates us to convert type-1 edge retrieval to range reporting on 2D points (introduced in Section~\ref{sec:pre}). Towards the purpose, create a set $P$ of points, which has a point $(A_u, A_{w^*})$ for every $\set{u,v}$ whose sentinel $w^*$ is not null. Attach edge $\set{u,v}$ to the point $(A_u, A_{w^*})$ so that the former can be fetched along with the latter. The size of $P$ is at most $m$. Given $q = [x_1, x_2]$, we can find all the type-1 edges by enumerating the points of $P$ inside the rectangle $[x_1, \infty) \times (-\infty, x_2]$. Hence, we can store $P$ in a Chazelle's structure (see Section~\ref{sec:pre}) that has $O(|P|) = O(m)$ space and ensures an $\tO(1)$ delay in reporting the type-1 edges of any $q$. 

\extraspacing {\bf Type 3.} A similar approach works for type 3. Let $\set{u,w}$ be an edge appearing in at least one occurrence $\tri(u,v,w)$ in $G$. It is a type-3 edge of $q = [x_1, x_2]$ if and only if $x_1 \le A_u$ and $A_w \le x_2$. By adapting the earlier discussion in a straightforward manner, we conclude that there is a structure of $O(m)$ space allowing us to retrieve all the type-3 edges with an $\tO(1)$ delay.

\extraspacing {\bf Listing $\bm{\Theta(m^*)}$ Triangles.} The above has explained how to retrieve $E^*_q$, but an RTE query still needs to report $\Theta(m^*)$ triangles. Next, we remedy the issue by slightly modifying our solution so far.

\vgap 

Recall that, in dealing with type 1, we attached the edge $\set{u,v}$ to the point $(A_u, A_{w^*})$ generated from the edge. Now, we attach $\tri(u,v,w^*)$ to $(A_u, A_{w^*})$ as well. This way, when $(A_u, A_{w^*})$ is found, we obtain both $\set{u,v}$ and $\tri(u,v,w^*)$ for free. After applying the same idea to type-2 and type-3, we can assert that, whenever the query algorithm finds a type-1, -2, or -3 edge, it must have also found a triangle in $G_q$. Therefore, the algorithm can report the triangles in $G_q$ with an $\tO(1)$ delay, although the same triangle may be reported up to three times\footnote{An occurrence $\tri(u,v,w)$ can be reported only when $\{u,v\}$, $\{v,w\}$, or $\{u,w\}$ is output as a type-1, -2, or -3 edge, respectively.}. By applying the duplicate-removal technique in Section~\ref{sec:pre}, we now have an algorithm that can enumerate $\Theta(m^*)$ distinct triangles with an $\tO(1)$ delay. The number of distinct triangles reported is at least $m^*/3$ and at most $3m^*$. 

\subsection{The Small-Delay Triangle Listing Problem} \label{sec:prob2-triangle:sdtl} 

In this subsection, we will concentrate on a standalone problem defined as follows.

\minipg{\minipgwidth}{
	{\bf Small-Delay Triangle Listing (SDTL).}
	$G$ is an undirected graph with $m$ edges, each of which appears in at least one triangle. We are given $\Omega(m)$ {\em free triangles} and $O(m)$ {\em forbidden triangles}. Design an algorithm to enumerate all the triangles of $G$ --- except for the forbidden ones --- with a small delay (free triangles must be enumerated). No preprocessing is allowed. 
}

\noindent We will settle the problem with an algorithm of delay $\tO(m^{\fr{\omega-1}{\omega+1}})$. 

\vgap

Suppose that $G$ has $\out$ triangles in total. Our starting point is an algorithm of Bjorklund et al.\ \cite{bpwz14} which is able to list $k$ triangles in $\alpha \cdot m^{\fr{3(\omega-1)}{\omega+1}} k^{\fr{3-\omega}{\omega+1}}$ time, where $\alpha = \tO(1)$, for a parameter $k \in [\Omega(m), \out]$. As far as the algorithm of \cite{bpwz14} is concerned, we can consider $\out$ known because it can be found in $O(m^{2\omega/(\omega+1)})$ time \cite{ayz97} which is $O(m^{\fr{3(\omega-1)}{\omega+1}} k^{\fr{3-\omega}{\omega+1}})$. The algorithm of Bjorklund et al.\ does not have a small delay, but we will turn it into one that does. 

\vgap 

We run the algorithm of Bjorklund et al.\ \cite{bpwz14} with geometrically-increasing $k$ and, in each run, report only some, but not all, of the triangles. How many triangles are reported in each run is decided strategically to keep the delay small. Let $S^0_\no$ be the set of forbidden triangles and $S^0_\yes$ the set of free triangles in the beginning. Set $k_0 = |S^0_\no| + |S^0_\yes|$. When running the algorithm of \cite{bpwz14} for the $i$-th  time, we set its parameter $k$ to $k_i = \min\{3^i k_0, \out\}$. We enforce the invariant that, when run $i$ starts, there are always a set $S^{i-1}_\no$ of forbidden triangles and a set $S^{i-1}_\yes$ of free triangles. The set $S^{i-1}_\yes$ will be reported with a small delay during the $i$-th run (details to be clarified shortly). 

\vgap 

Specifically, suppose that the $i$-th run finds a set $S^i_\raw$ of $k_i$ triangles (some of which have been output in previous runs). We generate the forbidden and free sets for the next run as follows:
\myeqn{
	S^i_\no = S^{i-1}_\no \cup S^{i-1}_\yes\text{ and then } 
	S^i_\yes &=& S^i_\raw \setminus S^{i}_\no. \nn 
}
Run $i$ finishes in $\alpha \cdot m^{\fr{3(\omega-1)}{\omega+1}} k_i^{\fr{3-\omega}{\omega+1}}$ time. We instruct the run to output a triangle from $S^{i-1}_\yes$ every 
\myeqn{
	\alpha \cdot \fr{m^{\fr{3(\omega-1)}{\omega+1}} k_i^{\fr{3-\omega}{\omega+1}}}{|S^{i-1}_\yes|} 
	\label{eqn:prob2:triangle:delay}
}
atomic operations. We will show $|S^{i-1}_\yes| = \Omega(k_i)$, with which the delay in \eqref{eqn:prob2:triangle:delay} can be bounded as: 
\myeqn{
	\tO\left(\fr{m^{\fr{3(\omega-1)}{\omega+1}}}{k_i^{\fr{2\omega-2}{\omega+1}}} \right)
	=
	\tO\left(m^{\fr{\omega-1}{\omega+1}} \right) \label{eqn:sdtl-delay}
}
where the equality used $k_i \ge k_0 = \Omega(m)$. 

\vgap 

For $i = 1$, $|S^{i-1}_\yes| = \Omega(k_0)$ follows directly from the definition of the SDTL problem (i.e., we have $\Omega(m)$ free triangles to start with). To prove $|S^{i-1}_\yes| = \Omega(k_i)$ for $i \ge 2$, we derive: 
\myeqn{
	|S^{i-1}_\no|
	\le 
	|S^0_\no| + |S^0_\yes| + \sum_{j=1}^{i-2} |S^j_\raw| 
	= k_0 + \sum_{j=1}^{i-2} 3^j \cdot k_0 
	=
	\sum_{j=0}^{i-2} 3^j \cdot k_0 
	< \fr{3^{i-1} k_0}{2}.  \nn 
}
Therefore:
\myeqn{
	|S^{i-1}_\yes| &\ge& |S^{i-1}_\raw| - |S^{i-1}_\no| 
	> k_{i-1} - 3^{i-1}k_0/2 
	= 
	k_{i-1}/2 
	= \Omega(k_i). \nn
}
We now conclude that the delay of our algorithm is as given in \eqref{eqn:sdtl-delay}.

\subsection{Proof of Theorem~\ref{thm:prob2:triangle}} \label{sec:prob2-triangle:tri} 

We are ready to explain how to solve Problem 2 with $Q =$ triangle. In preprocessing, we build an RTE structure (Section~\ref{sec:prob2-triangle:range-tri-edges}) on $G$. Now, consider a (Problem-2) query with interval $q$. We start by issuing an RTE query to retrieve $E^*_q$, i.e., the set of edges appearing in at least one triangle of $G_q$. This, in effect, generates $G^*_q$, which is the subgraph of $G_q$ induced by the edges in $E^*_q$. In addition, the RTE query has also enumerated a set $S$ of $\Theta(m^*)$ triangles in $G_q$, where $m^* = |E^*_q|$. The size of $S$ falls in $[\fr{m^*}{3}, 3m^*]$. 

\vgap 

Our remaining mission is to enumerate the triangles in $G^*_q$ that are outside $S$. Note that $G^*_q$ is a graph with $m^*$ edges and at least $\Theta(m^*)$ triangles. This motivates us to convert the mission to the SDTL problem, which has been solved in Section~\ref{sec:prob2-triangle:sdtl}. However, the SDTL problem requires $\Theta(m^*)$ free triangles and $O(m^*)$ forbidden triangles as part of the input. Unfortunately, we do not seem to have these triangles at the moment. 

\vgap 

We overcome this obstacle by, interestingly, dividing $S$ into $S_\yes$ and $S_\no$,  such that $S_\yes$ (resp.\ $S_\no$) serves as the set of free (resp.\ forbidden) triangles. Recall that the RTE query algorithm, denoted as $\A$, is designed to enumerate an edge in $E^*_q$ with a delay $\Delta = \tO(1)$ and a triangle in $S$ also with a delay $\Delta$. Therefore, it must finish within $t_\mit{max} = \max\{\Delta \cdot (|E^*_q|+1), \Delta \cdot (|S|+1)\} \le \Delta \cdot (3m^*+1)$ time. We can now apply the buffering technique in Section~\ref{sec:pre} with $\alpha = 18$ to turn $\A$ into an algorithm that outputs a triangle at the end of each epoch, which has a length $18\Delta$. The total number of epochs is at most $\fr{t_{max}}{18 \Delta} \le \fr{3m^*+1}{18}$. Thus, when $\A$ finishes, we have output at most $(3m^*+1)/18$ triangles, whereas the buffer $B$ (defined in Section~\ref{sec:pre}) still has at least $|S| - \fr{3m^* + 1}{18} = \Theta(m^*)$ triangles. We can, thus, set $S_\yes$ to the content of $B$ when $\A$ finishes, and $S_\no$ to the set of triangles already output.  

\vgap 

We can now apply the SDTL algorithm on $G^*_q$ and, thus, complete the proof of Theorem~\ref{thm:prob2:triangle}.

\section{Problem 2: Near-Constant Delays} \label{sec:prob2-constant} 

This section will focus on two instances of Problem 2 where it is possible to achieve $\tO(1)$ delays with space substantially smaller than Theorem~\ref{thm:prob2:gen}. We will discuss first $Q =$ $\ell$-star in Section~\ref{sec:prob2-constant:star} and then $Q=$ $2\ell$-cycle in Section~\ref{sec:prob2-constant:cycle}. We will focus on explaining how to enumerate a perhaps-not-distinct occurrence with an $\tO(1)$ delay, while ensuring each occurrence to be output only a constant number of times. Owing to the duplicate-removal method in Section~\ref{sec:pre}, we can modify the algorithms to enumerate only {\em distinct} occurrences with $\tO(1)$ delays. 

\subsection{$\bm{\ell}$-Stars} \label{sec:prob2-constant:star} 

Recall that an $\ell$-star is a tree with only one non-leaf node, which we will refer to as the star's {\em center}. Consider a query with interval $q$. We refer to a node $u$ as a {\em $q$-center} if $G_q$ has at least one $\ell$-star occurrence with $u$ as the center. Once $u$ is found, it becomes a trivial matter to enumerate all the $\ell$-stars having $u$ as the center with an $\tO(1)$ delay. Specifically, we can first (use a binary search tree to) retrieve all the neighbors $v$ of $u$ in $G$ satisfying $A_v \in q$. From those neighbors, any $\ell$ distinct vertices form an $\ell$-star together with $u$ (as the center). It is rudimentary to ensure an $\tO(1)$ delay in enumerating all those stars. 

\vgap 

Next, we concentrate on designing a structure to enumerate the $q$-centers with an $\tO(1)$ delay. Consider an arbitrary $\ell$-star in $G$ with center $u$. Sort the star's $\ell + 1$ vertices in ascending order of attribute and look for the position of $u$. If $u$ is the $r$-th smallest, we will refer to the star as a {\em rank-$r$ $\ell$-star} and $u$ as a {\em rank-$r$ $q$-center}. 


\vgap 

Now, fix an $r \in [1, \ell+1]$. We will describe a structure to support the following operation:

\vgap 

\minipg{\minipgwidth}{
	Given an interval $q$, find all the rank-$r$ $q$-centers, i.e., all vertices $u \in V$ s.t.\ $G_q$ has a rank-$r$ $\ell$-star with $u$ as the center.
}

\vgap

\noindent Consider any rank-$r$ $\ell$-star in $G$ having $u$ as the center. Let us write out the star's vertices as $v_1, ..., v_{r-1}, u,$ $v_{r+1}, ..., v_\ell$ in ascending order of attribute. For a $q = [x_1, x_2]$, the $\ell$-star appears in $G_q$ if and only if $x_1 \le A_{v_1}$ and $A_{v_\ell} \le x_2$. Refer to $v_1$ as a {\em left $r$-sentinel} of $u$ and to $v_\ell$ as a {\em right $r$-sentinel} of $u$. From all the left $r$-sentinels of $u$ (one from each rank-$r$ $\ell$-star with center $u$), identify the one $v^*_1$ with the largest attribute. Similarly, from all the right $r$-sentinels of $u$, identify the one $v^*_\ell$ with the smallest attribute. Observe that $u$ is a rank-$r$ $q$-center if and only if $x_1 \le A_{v^*_1}$ and $A_{v^*_\ell} \le x_2$. We can therefore convert the retrieval of rank-$r$ $q$-centers into range reporting on 2D points (review Section~\ref{sec:pre}), in the same way as illustrated in Section~\ref{sec:prob2-triangle:range-tri-edges}. Following Section~\ref{sec:pre}, we can create a Chazelle's structure on $n$ points --- each point created for a vertex $u \in V$ in the way explained --- that has $O(n)$ space and, given any $q$, can list the rank-$r$ $q$-centers with an $\tO(1)$ delay. This completes the proof of Theorem~\ref{thm:prob2:l-star}. 

\subsection{$\bm{2\ell}$-Cycles} \label{sec:prob2-constant:cycle} 

We will start with an assumption: all queries specify a fixed $q = (-\infty, \infty)$, namely, there is effectively only one query, which enumerates all the $2\ell$-cycles in $G$. The assumption allows us to explain the core ideas with the minimum technical details and will be removed eventually.


\extraspacing {\bf Queries with $\bm{q = (-\infty, \infty)}$.} Given a $2\ell$-cycle occurrence, we refer to the vertex $u$ in the cycle having the smallest attribute as the occurrence's {\em anchor}. Let $v$ be the vertex in the cycle such that cutting the cycle at $u$ and $v$ gives two $\ell$-paths connecting $u$ and $v$. We will refer to $v$ as the occurrence's {\em inverse anchor}, the pair $(u,v)$ as an {\em anchor pair}, and the two aforementioned paths as {\em cycle $\ell$-paths}. The number of cycle $\ell$-paths is at most $\#P_\ell$ (recall that $\#P_\ell$ is the total number of $\ell$-paths).


\vgap 

The problem may appear deceivingly simple: can't we answer a query by simply concatenating, for each anchor pair $(u,v)$, every two cycle $\ell$-paths from $u$ to $v$? This does not work because the two cycle $\ell$-paths may share common vertices other than $u$ and $v$, in which case the concatenation does not yield a $2\ell$-cycle! This motivates a crucial notion: two cycle $\ell$-paths are {\em interior disjoint} if they (i) have the same anchor pair $(u,v)$, and (ii) do not share any common vertex except $u$ and $v$. Concatenating two cycle $\ell$-paths from $u$ to $v$ spawns a $2\ell$-cycle if and only if those paths are interior disjoint. The challenge we are facing at this moment is the following problem. 

\minipg{\minipgwidth}{
	Design a structure to support the following operation: given a cycle $\ell$-path $\pi$ from anchor $u$ to inverse anchor $v$, list all the cycle $\ell$-paths interior disjoint with $\pi$ with an $\tO(1)$ delay.
}

\noindent We will overcome the challenge with a structure of $\tO(\#P_\ell)$ space. 

\vgap 

Our main observation is that the operation can be converted to range reporting on $(\ell-1)$-dimensional points (review Section~\ref{sec:pre}). To explain, let us consider any cycle $\ell$-path $\pi$ from anchor $u$ to inverse anchor $v$. After excluding $u$ and $v$, the path has $\ell-1$ vertices, which we list as $w_1, w_2, ..., w_{\ell-1}$ in ascending order of attribute\footnote{The order should not be confused with the order by which the vertices  appear in $\pi$.}. Convert $\pi$ into an $(\ell-1)$-dimensional point $(A_{w_1}, ..., A_{w_{\ell-1}})$. Let $P_{u,v}$ be the set of points thus obtained from all the cycle $\ell$-paths with $(u,v)$ as the anchor pair. 

\vgap

Now, consider another cycle $\ell$-path $\pi'$ from $u$ to $v$. List the vertices of $\pi'$ other than $u$ and $v$ as $w'_1, w'_2, ..., w'_{\ell-1}$ also in ascending order of attribute. If $\pi'$ is interior disjoint with $\pi$, each $A_{w'_i}$ ($i \in [1, \ell-1]$) must fall in one of the $\ell$ open intervals: 
\myeqn{
	(-\infty, A_{w_1}), (A_{w_1}, A_{w_2}), ..., (A_{w_{\ell-2}}, A_{w_{\ell-1}}), (A_{w_{\ell-1}}, \infty). \label{eqn:prob2-constant:cycle:intvs} 
}
Therefore, $(A_{w_1'}, ..., A_{w_{\ell-1}'})$ --- the point converted from $\pi'$ --- must fall in one of the following $\ell^{\ell-1} = O(1)$ rectangles: $
	q_1 \times q_2 \times ... \times q_{\ell-1}, 
$
where each $q_i$ ($i \in [1, \ell-1]$) is taken independently from one of the intervals in \eqref{eqn:prob2-constant:cycle:intvs}. As per Section~\ref{sec:pre}, by creating a range tree on $P_{u,v}$ of $\tO(|P_{u,v}|)$ space, we can enumerate all the points in such a rectangle with an $\tO(1)$ delay. 

\vgap

The conclusion from the above is that, for each anchor pair $(u,v)$, we can create a range tree of $\tO(|P_{u,v}|)$ space which, given any cycle $\ell$-path cycle $\pi$ from $u$ to $v$, permits the enumeration of every cycle $\ell$-path $\pi'$, which is interior disjoint with $\pi$, with an $\tO(1)$ delay. The structures of all the anchor pairs use in total $\sum_{\text{anc.\ pair $(u,v)$}} \tO(|P_{u,v}|) = \tO(\#P_\ell)$ space. 

\vgap 

With the challenge conquered, listing all the $2\ell$-cycles becomes an easy matter. We simply look at each cycle $\ell$-path $\pi$, retrieve every $\ell$-path $\pi'$ interior disjoint with $\pi$, and make a cycle by concatenating $\pi$ and $\pi'$. The delay in cycle reporting is $\tO(1)$ (each $2\ell$-cycle can be reported twice). 

\extraspacing {\bf Arbitrary Queries.} Next, we remove the constraint $q = (-\infty, \infty)$ and tackle queries with arbitrary $q$. A new issue now arises: a query can no longer afford to look at all the cycle $\ell$-paths. We say that a cycle $\ell$-path from anchor $u$ to inverse anchor $v$ {\em contributes} to $G_q$ if it makes a $2\ell$-cycle in $G_q$  with another interior disjoint cycle $\ell$-path. We need a way to list only the contributing cycle $\ell$-paths.

\vgap 

Fix any cycle $\ell$-path $\pi$ with anchor pair $(u,v)$. Let $S_\pi$ be the set of $2\ell$-cycles in $G$ that include $\pi$ and have $(u,v)$ as the anchor pair. Take an arbitrary cycle from $S_\pi$. By definition of anchor, $u$ has the smallest attribute among the cycle's vertices. Let $w$ be the vertex in the cycle with the largest attribute. For $q = [x_1, x_2]$, the cycle appears in $G_q$ if and only if $x_1 \le A_u$ and $A_w \le x_2$. Let $w^*$ be the vertex with the smallest attribute among all such $w$'s. It becomes evident that $\pi$ contributes to the $G_q$ of $q = [x_1,x_2]$ if and only if $x_1 \le A_u$ and $A_{w^*} \le x_2$. We can therefore convert the retrieval of contributing cycle $\ell$-paths to range reporting on 2D points, using the method in Section~\ref{sec:prob2-triangle:range-tri-edges}. The resulting structure (a Chazelle's structure) stores a point converted from every cycle $\ell$-path and uses $O(\#P_\ell)$ space. Give any $q$, we can  list the cycle $\ell$-paths contributing to $G_q$ with an $\tO(1)$ delay. 

\vgap

Suppose that we have found a contributing cycle $\ell$-path $\pi$ with anchor pair $(u,v)$. As before, we proceed to find the cycle $\ell$-paths $\pi'$ interior disjoint with $\pi$. The new requirement here, however, is that $\pi'$ needs to be contributing as well. Recall that, in the $q = (-\infty, \infty)$ scenario, we converted the task to range reporting on $(\ell-1)$-dimensional points. To deal with arbitrary $q = [x_1, x_2]$, we will increase the dimension by one. 

\vgap 

To explain, in a fashion like before, let us list out the vertices of $\pi$ --- after excluding $u$ and $v$ --- as $w_1, ..., w_{\ell - 1}$ in ascending order of attribute. Denote by $w_{\text{max}}$ the vertex in $\pi$ with the largest attribute ($w_{\text{max}}$ can be $v$). Convert $\pi$ to an $\ell$-dimensional point $(A_{w_1}, ..., A_{w_{\ell-1}},$ $A_{w_{\text{max}}})$. Let $(A_{w_1'}, ..., A_{w'_{\ell-1}}, A_{w'_{\text{max}}})$ be the point converted from $\pi'$ in the same manner. As we already know $A_u \in [x_1, x_2]$ (recall that $\pi$ is a contributing path), $\pi'$ is a path we want if and only if it satisfies the conditions below:
\myitems{
	\item $A_{w'_i}$ ($1 \le i \le \ell-1$) falls in one of the $\ell$ intervals in \eqref{eqn:prob2-constant:cycle:intvs}; 
	\item $A_{w'_{\text{max}}} \le x_2$. 
}
\noindent Thus, the point $(A_{w_1'}, ..., A_{w'_{\ell-1}}, A_{w'_{\text{max}}})$ must fall in one of the following $\ell^{\ell-1} = O(1)$ rectangles: 
$
	q_1 \times q_2 \times ... \times q_{\ell-1} \times (-\infty, x_2], 
$
where each $q_i$ ($i \in [1, \ell-1]$) is an interval taken independently from \eqref{eqn:prob2-constant:cycle:intvs}. By the above reasoning, for each anchor pair $(u,v)$, we create a set $P_{u,v}$ of $\ell$-dimensional points, each converted from a cycle $\ell$-path with anchor pair $(u,v)$, and then build a range tree on $P_{u,v}$. The range trees of all anchor pairs use  $\sum_{\text{anc.\ pair $(u,v)$}} \tO(|P_{u,v}|) = \tO(\#P_\ell)$ space in total.

\vgap 

We now elaborate on the overall algorithm for answering a (Problem-2) query with parameter $q$. First, enumerate all the cycle $\ell$-paths contributing to $G_q$ with an $\tO(1)$ delay; call this the {\em outer enumeration}. Every time such a path $\pi$ --- say with anchor pair $(u,v)$ --- is obtained, we suspend outer enumeration and utilize the range tree on $P_{u,v}$ to find all the paths $\pi'$ discussed previously with an $\tO(1)$ delay. Upon the delivery of a $\pi'$, concatenate it with $\pi$ and output the $2\ell$-cycle obtained. After exhausting all such $\pi'$, we resume outer enumeration. This concludes the proof of Theorem~\ref{thm:prob2:cycle}.

\appendix 

\section*{Appendix} 

\section{Correctness of the Reduction in Section~\ref{sec:prob1-lb-ub:lb}} \label{app:prob1-lb:correctness}

In our construction, $S_i$ ($i \in [1,s]$) corresponds to two set vertices with attribute values $i$ and $i + s$, respectively. To facilitate derivation, we make a copy of each set: define $S_i = S_{i-s}$ for each $i \in [s+1, 2s]$. In the rest of the proof, we hold the view that each $S_i$ ($i \in [1,2s]$) corresponds to only one set vertex, the one with attribute value $i$. 

\vgap

Consider a wedge occurrence with vertices $u, v$, and $w$ where the edges are $\set{u,v}$ and $\set{v, w}$. We classify it as one of the two types below: 
\myitems{
	\item (type e-s-e) $u$ and $w$ are element vertices and $v$ is a set vertex; 
	\item (type s-e-s) $u$ and $w$ are set vertices and $v$ an element vertex. 
}

\begin{lemma} \label{lmm:lb-wedge:num-each-type}
	For any interval $q = [x, y]$ satisfying $1 \leq x < s + 1/2 < y\leq 2s$, we have
	\myitems{
		\item the number of e-s-e wedges in $G_q$ is $\sum_{i\in [x, y]}{|S_i|\choose 2}$;
		\item the number of s-e-s wedges in $G_q$ is $\sum_{i\in [x, y], j \in [i+1, y]}|S_i \cap S_j|$.
	} 
\end{lemma}
\begin{proof}
	To prove the first bullet, define an {\em e-s-e tuple} as $(e_1, S_i, e_2)$ where $i \in q$ and $e_1$ and $e_2$ are distinct elements in $S_i$. The number of such tuples is $\sum_{i=x}^{y} \binom{|S_i|}{2}$. Our construction ensures a one-one correspondence between e-s-e tuples and e-s-e wedges in $G_q$.  
	
	\vgap
	
	To prove the second bullet, define an {\em s-e-s tuple} as $(S_i, e, S_j)$ where $x \le i < j \le y$ and $e \in S_i \cap S_j$. The number of such tuples is $\sum_{i\in [x, y], j \in [i+1, y]}|S_i \cap S_j|$. Our construction ensures a one-one correspondence between s-e-s tuples and s-e-s wedges in $G_q$.
\end{proof}

To find out whether $S_a \cap S_b$ is empty, our reduction issues four Problem-1 queries with intervals $q_1 = [a, s+b]$, $q_2 = [a+1, s+b]$, $q_3 = [a, s+b-1]$, and $q_4 = [a+1, s+b-1]$, respectively. The above lemma is applicable to all these intervals. For $i \in [1, 4]$, let $c_i'$ (resp.\ $c_i''$) be the number of e-s-e (resp.\ s-e-s) wedges in $G_{q_i}$; this means that $c_i$, the total number of wedges in $G_{q_i}$, equals $c_i' + c_i''$. According to Lemma~\ref{lmm:lb-wedge:num-each-type}, we have:
\myeqn{
	&& c_1' - c_2' - c_3' + c_4' \nn \\
	&=& \sum_{i\in [a, s+b]}{|S_i|\choose 2} - \sum_{i\in [a+1, s+b]}{|S_i|\choose 2} 
	- \sum_{i\in [a, s+b-1]}{|S_i|\choose 2} 
	+ \sum_{i\in [a+1, s+b-1]}{|S_i|\choose 2} \nn \\
	&=& 0\nn
}
and 
\myeqn{
	&& c_1'' - c_2'' - c_3'' + c_4'' \nn \\
	&=& \left( \sum_{\substack{i \in [a, s+b]\\j \in [i+1,s+b]}}|S_i \cap S_j| - 
	\sum_{\substack{i \in [a+1, s+b]\\ j \in [i+1, s+b]}}|S_i \cap S_j| \right) - \nn \\
	&& \left(\sum_{\substack{i \in [a, s+b-1]\\ j \in [i+1, s+b-1]}}|S_i \cap S_j| - 
	\sum_{\substack{i \in [a+1, s+b-1]\\ j \in [i+1, s+b-1]}}|S_i \cap S_j| \right) \nn 
}
\myeqn{
	&=& \sum_{j \in [a+1, s+b]} |S_a \cap S_j| - \sum_{j\in [a+1, s+b-1]}|S_a \cap S_j| \nn \\
	&=& |S_a \cap S_{s+b}| \nn \\ 
	&=& |S_a \cap S_b|. \nn
}
We thus conclude that $c_1 - c_2 - c_3 + c_4 = |S_a \cap S_b|$.

\section{Proof of Lemma~\ref{lmm:prob1-wedge:color}} \label{app:prob1-wedge:color}

Let us first consider a variant of the set disjointness problem. 

\minipg{\minipgwidth}{
	{\bf Weighted Set Intersection Size.} We have $s \ge 2$ sets $S_1,$ $S_2,$ ..., $S_s$. Each $S_i$ ($i \in [1, s]$) is associated with a function $\weight_{S_i}$ which assigns to each element $e \in S_i$ a value $\weight_{S_i}(e)$. Given distinct set ids $a, b \in [1,s]$, a query returns 
	\myeqn{
		\size(S_a, S_b) 
		= 
		\sum_{e \in S_a \cap S_b} \weight_{S_a}(e) \cdot \weight_{S_b}(e). 
		\label{eqn:prob1-wedge:color:intr-weight} 
	}
}

Let $N = \sum_{i=1}^s |S_i|$. For any $\lambda \in [1, \sqrt{N}]$, it is straightforward to build a structure of $O(N^2 / \lambda^2)$ space answering a query in $O(\lambda)$ time. Call $S_i$ ($i \in [1, s]$) a {\em large} set if $|S_i| > \lambda$, or a {\em small} set otherwise. The number of large sets is at most $N / \lambda$. For each pair $(i, j) \in [1, s] \times [1, s]$, $i \ne j$, such that $S_i$ and $S_j$ are both large, we store $ \size(S_i, S_j)$; the space needed is $O(N^2/\lambda^2)$. Given a query with parameters $a$ and $b$, return $ \size(S_a, S_b)$ directly if $S_a$ and $S_b$ are both large. Otherwise, assume, w.l.o.g., that $S_a$ is small. We compute $S_a \cap S_b$ in $O(\lambda)$ time using a hash table (for each $e \in S_a$, check if $e \in S_b$). The result $\size(S_a, S_b)$ can then be obtained easily. 

\vgap 

Equipped with the above, next we describe a structure for the colored range wedge counting problem to prove Lemma~\ref{lmm:prob1-wedge:color}. 

\extraspacing {\bf Structure.} First obtain a canonical collection $\C$ of $V$ (defined in Section \ref{sec:prob1-wedge}) satisfying $\sum_{U \in \C} |U| = \tO(n)$. For each $U \in \C$ --- recall that $U$ is a subset of $V$ --- construct a weighted set as follows:
\myitems{
	\item $S_U =$ the set of black vertices adjacent to at least one vertex in $U$;
	\item for each $b \in S_U$, $\weight_{S_U}(b) =$ the number of vertices in $U$ adjacent to  $b$. 
}
These weighted sets constitute an instance of the weighted set intersection size problem. Build a structure described earlier on the instance using the given parameter $\lambda$. The lemma below implies that the structure occupies $\tO(m^2/\lambda^2)$ space. 

\begin{lemma} \label{lmm:prob1-wedge:color:space}
	$\sum_{U \in \C} |S_U| = \tO(m)$. 
\end{lemma}

\begin{proof}
	Each $b \in S_U$ is adjacent to a vertex $u \in U$. Pay a dollar to the edge $\{b, u\}$  for each such pair $(b, u)$. Since an edge can receive a dollar only if it has a vertex in $U$, it can receive up to two dollars\footnote{Two is possible: this happens when $b$ and $u$ are both black and both appear in $U$.}. $|S_U|$ is no more than the number of dollars paid. Do the above for all $U \in \C$. Each edge in $G$ can receive $\tO(1)$ dollars in total because every vertex appears in $\tO(1)$ subsets in $\C$ (Property P4-1 of $\C$; see Section~\ref{sec:prob1-wedge}). 
\end{proof}

For any distinct $U, U' \in \C$, define $\size(S_{U}, S_{U'})$ as in \eqref{eqn:prob1-wedge:color:intr-weight}. On the other hand, for each $U \in \C$, define 
\myeqn{
	\size(S_U, S_U) 
	&=&
	\sum_{b \in S_U} {\weight_{S_U}(b) \choose 2}. \nn 
}
We store the value $\size(S_U,S_U)$ for all $U$. The total space is $\tO(m^2/\lambda^2)$.

\vgap 

Before proceeding, the reader should note the following subtle fact about the function $\size(.,.)$:

\minipg{\minipgwidth}{
	{\bf Fact B-1:} $\size(S_{U},S_{U'})$ is the number of occurrences $\wedge(u,v,w)$ in $G$ such that $u \in S_{U}$, $w \in S_{U'}$, and $v$ is black. 
}

\noindent The fact holds even if $U = U'$.

\extraspacing {\bf Query.} Given a query with interval $q$, in $\tO(1)$ time we can pick $h = \tO(1)$ members $U_1, ...,U_h$ from $\C$ that form a partition of $V_q$ (Property P4-2 of $\C$). The query returns
\myeqn{
	\sum_{i, j \in [1, h]: i \le j} \size(S_{U_i}, S_{U_j}). 
	\label{eqn:prob1-wedge:color:qry}
}
Each $\size(S_{U_i}, S_{U_j})$ is either explicitly stored or can be obtained from the weighted set intersection size structure in $O(\lambda)$ time. The overall query time is therefore $\tO(\lambda)$.

\vgap 

Fact B-1 and $U_1, ...,U_h$ forming a partition of $V_q$ assure us that \eqref{eqn:prob1-wedge:color:qry} counts only occurrences $\wedge(u,v,w)$ in $G$ such that $A_u \in q$, $A_w \in q$, and $v$ is black. To complete the correctness argument, we still need to show that \eqref{eqn:prob1-wedge:color:qry} counts every such occurrence exactly once. Indeed, there exist unique $a, b \in [1, h]$ such that $a \le b$, $u \in U_a$, and $w \in U_b$. The wedge is counted only by the term in \eqref{eqn:prob1-wedge:color:qry} with $i = a$ and $j = b$.
\section{Proof of Lemma~\ref{lmm:prob2-arbitrary:agm-gen}} \label{app:prob2-arbitrary:agm-gen} 


Let us first review H\"older's Inequality. Fix some positive integers $\alpha$ and $\beta$. Let
\myitems{
	\item $x_{i, j}$, for each $i \in [1, \alpha]$ and $j\in [1, \beta]$, be non-negative real numbers; 
	\item $y_j$, for each $j \in [1, \beta]$, be non-negative real numbers satisfying $\sum_{j=1}^\beta y_j \geq 1$.
}
Under the convention $0^0 = 0$,  H\"older's inequality states that:
\myeqn{
	\sum_{i=1}^{\alpha}	\prod_{j=1}^{\beta} x_{i, j}^{y_j} \leq 	\prod_{j=1}^{\beta} \left(\sum_{i=1}^{\alpha} x_{i, j}\right)^{y_j}.	\label{eqn:holder}
}
A proof can be found in \cite{f04}. 


\extraspacing 

We now 
return to the context of Lemma~\ref{lmm:prob2-arbitrary:agm-gen}.
Given any $j\in [1, d-1]$ and $(I_1, I_2, ..., I_j)\in \I_1 \times ... \I_j$, we will use $B(I_1, I_2, ..., I_j)$ as a short-form for the $d$-dimensional box
\myeqn{
	B(I_1, ..., I_j, \dom(X_{j+1}), ..., \dom(X_d)). \nn
}
As a special case, define $B(\emptyset) = B(\dom(X_1)$, $..., \dom(X_d))$.

\begin{lemma} \label{lmm:prob2-arbitrary:agm-gen:help} 
	For any $j\in [1, d]$, we have
	\myeqn{
		\sum_{I_j \in \I_j} \prod_{e \in \E} |R_e \ltimes B(I_1, ..., I_j)|^{W(e)} 
		\le
		\prod_{e \in \E} |R_e \ltimes B(I_1, ..., I_{j-1})|^{W(e)}. \nn
	}
\end{lemma}

\begin{proof}
	Define 
	\myeqn{
		\E_j = \{e\in \E \mid X_j \in e \}. \nn
	}
	Since $\sum_{e\in \E_j}W(e) \geq 1$ ($W$ is a fractional edge covering), from H\"older's inequality \eqref{eqn:holder} we have
	\myeqn{
		&& \sum_{I_j \in \I_j}\prod_{e\in \E_j}|R_e \ltimes B(I_1, ..., I_j)|^{W(e)}  \nn \\
		&\le& \prod_{e\in \E_j} 
		\Big(\sum_{I_j \in \I_j}|R_e \ltimes B(I_1, ..., I_j)|\Big)^{W(e)}
		\nn \\
		&\le& 
		\prod_{e\in \E_j} \Big|R_e \ltimes B\big(I_1, ..., I_{j-1}, \dom(X_j) \big)\Big|^{W(e)} \nn \\
		&=& 
		\prod_{e \in \E_j} |R_e \ltimes B(I_1, ..., I_{j-1})|^{W(e)}  
		\label{eqn:prob2-arbitrary:agm-gen:decom1}
	}
	where the second inequality used the fact that $\I_j$ is a set of disjoint intervals in $\dom(X_j)$.
	
	\vgap 
	
	For each $e \in \E \setminus \E_j$, $R_e \ltimes B(I_1, ..., I_j)$ does not depend on $I_j$ and can be rewritten as $R_e \ltimes B(I_1, ..., I_{j-1})$. We can thus derive: 
	\myeqn{
		&& \hspace{-7mm} \sum_{I_j \in \I_j} \prod_{e \in \E} |R_e \ltimes B(I_1, ..., I_j)|^{W(e)} 
		\nn \\
		&& 
		\hspace{-10mm} =
		\sum_{I_j \in \I_j} 
		\Big( \prod_{e \in \E\setminus \E_j} |R_e \ltimes B(I_1, ..., I_{j})|^{W(e)} \cdot  
		\prod_{e\in \E_j}|R_e \ltimes B(I_1, ..., I_j)|^{W(e)} \Big) \nn \\
		&& 
		\hspace{-10mm} =
		\prod_{e \in \E\setminus \E_j} |R_e \ltimes B(I_1, ..., I_{j})|^{W(e)} 
		\cdot
		\sum_{I_j \in \I_j}\prod_{e\in \E_j}|R_e \ltimes B(I_1, ..., I_j)|^{W(e)}  \nn \\
		&& 
		\hspace{-10mm} \le
		\prod_{e \in \E\setminus \E_j} |R_e \ltimes B(I_1, ..., I_{j-1})|^{W(e)} \cdot  
		\prod_{e \in \E_j} |R_e \ltimes B(I_1, ..., I_{j-1})|^{W(e)} 
		\nn \\ 
		&& 
		\hspace{-10mm} =
		\prod_{e \in \E} |R_e \ltimes B(I_1, ..., I_{j-1})|^{W(e)}. \nn
	}
	where the inequality used \eqref{eqn:prob2-arbitrary:agm-gen:decom1}.
\end{proof}

We can prove Lemma~\ref{lmm:prob2-arbitrary:agm-gen} with $d$ applications of Lemma~\ref{lmm:prob2-arbitrary:agm-gen:help}:
\myeqn{
	&& \sum_{I_1 \in \I_1}...\sum_{I_d \in \I_d}
	\prod_{e \in \E} |R_e \ltimes B(I_1, ..., I_d)|^{W(e)} \nn \\
	&\leq & 
	\sum_{I_1 \in \I_1}...\sum_{I_{d-1} \in \I_{d-1}}
	\prod_{e \in \E} |R_e \ltimes B(I_1, ..., I_{d-1})|^{W(e)} \nn \\
	&\leq & 
	\sum_{I_1 \in \I_1}...\sum_{I_{d-2} \in \I_{d-2}}
	\prod_{e \in \E} |R_e \ltimes B(I_1, ..., I_{d-2})|^{W(e)} \nn \\
	&\leq& ... \nn \\
	&\leq & 
	\sum_{I_1 \in \I_1} 
	\prod_{e \in \E} |R_e \ltimes B(I_1)|^{W(e)} \nn \\
	&\le& 
	\prod_{e \in \E} |R_e|^{W(e)}. \nn
}

\section{Proof of Theorem~\ref{thm:prob2:gen}} \label{app:prob1-arbitrary} 


The reader should read this proof after having finished Section~\ref{sec:prob2-arbitrary}. The basic idea is to convert Problem 2 to range join. Let $\X$ (resp.\ $\E$) be the set of vertices (resp.\ edges) in the pattern graph $Q$. The reader should not confuse $\X$ and $\E$ with $V$ and $E$: the latter two are defined on the data graph $G$. For each edge $e \in \E$, construct a relation $R_e$ with two attributes by inserting, for each edge $\set{u,v}$ in $G$, two tuples $(u,v)$ and $(v,u)$. This defines a join instance $\R = \set{R_e \mid e \in \E}$ with input size $N = 2m \cdot |\E| = O(m)$. 

\vgap 

Every occurrence of $Q$ corresponds to a constant number of tuples in $\join(\R)$. Motivated by this, given a Problem-2 query with interval $q$, we issue a range join query on $\R$ with $q$, which guarantees retrieving all the occurrences. The issue, however, is that not every tuple in $\join(\R)$ gives rise to an occurrence. To see this, consider $Q =$ 4-cycle and, hence, $\R$ has four relations with schemes $(X_1, X_2)$, $(X_2, X_3)$, $(X_3, X_4)$, and $(X_4, X_1)$, respectively. Let $\set{u,v}$ be an arbitrary edge in $E$; tuples $(u,v)$, $(v,u)$, $(u,v)$, and $(v,u)$ exist in the four relations, respectively. Thus, $\join(\R)$ contains a tuple $(u,v,u,v)$ that does not correspond to any occurrence.

\vgap 

The issue can be eliminated by slightly modifying the structure of \cite{dk18}, which we review next. Consider an arbitrary set $\R$ of relations (with any number of attributes) defined in Section~\ref{sec:prob2-arbitrary}. Deep and Koutris \cite{dk18} proved the existence of a set $\B$ of boxes such that: 
\myitems{
	\item each box has the form $B(I_1, ..., I_d)$ where $I_i$ is an interval in $\dom(X_i)$ for $i \in [1, d]$; 
	\item the boxes are disjoint and their union is $B(\dom(X_1),$ $\dom(X_2),$ ...,$ \dom(X_d))$;
	\item for each box $B(I_1, ..., I_d)$, the join instance $\R_{I_1, ..., I_d}$ has a non-empty result;
	\item each box $B(I_1, ..., I_d)$ satisfies $\agm(I_1, ..., I_d) \le \Delta$; 
	\item $|\B| = O(N^{\rho^*}/\Delta)$.
}
The structure of \cite{dk18} simply stores $\B$ itself and uses $O(N^{\rho^*}/\Delta)$ space\footnote{Obviously, the relations of $\R$ also need to be stored separately.}. To enumerate $\join(\R)$, the algorithm of \cite{dk18} looks at each $B(I_1, ..., I_d) \in \B$ and applies a worst-case optimal join algorithm \cite{nprr18,nrr13,v14} to compute $\join(\R_{I_1, ..., I_d})$ in $\tO(\agm(I_1, ..., I_d)) = \tO(\Delta)$ time. This guarantees a delay of $\tO(\Delta)$. 

\vgap

We now adapt the structure to list all the occurrences of $Q$ in $G$ (fixing $q = (-\infty, \infty)$). Construct $\R$ from $G$ and $Q$ as before. Apply \cite{dk18} to find a set $\B$ with all the properties explained earlier. Then, inspect each box $B(I_1, ..., I_d) \in \B$ in turn and remove it from $\B$ if all the occurrences of $Q$ producible from $\join(\R_{I_1,...,I_d})$ can already be produced from the boxes inspected earlier. The size of $\B$ can only decrease and therefore is still bounded by $O(N^{\rho^*}/\Delta)$. To find the occurrences, apply a worst-case optimal join algorithm on each box in $\B$. As each box generates at least one new occurrence, we guarantee a delay of $\tO(\Delta)$. 

\vgap 

To support (Problem-2) queries with arbitrary $q$, use the adapted structure to replace that of \cite{dk18} in the solution presented in Section~\ref{sec:prob2-arbitrary:range-join}. All the analysis still holds through. We thus complete the proof of Theorem~\ref{thm:prob2:gen}.

\bibliographystyle{plainurl}
\bibliography{ref}

\begin{thebibliography}{10}

\bibitem{aw21}
Josh Alman and Virginia~Vassilevska Williams.
\newblock A refined laser method and faster matrix multiplication.
\newblock In {\em Proceedings of the Annual {ACM-SIAM} Symposium on Discrete
  Algorithms ({SODA})}, pages 522--539, 2021.

\bibitem{ayz95}
Noga Alon, Raphael Yuster, and Uri Zwick.
\newblock Color-coding.
\newblock {\em Journal of the ACM ({JACM})}, 42(4):844--856, 1995.

\bibitem{ayz97}
Noga Alon, Raphael Yuster, and Uri Zwick.
\newblock Finding and counting given length cycles.
\newblock {\em Algorithmica}, 17(3):209--223, 1997.

\bibitem{agm13}
Albert Atserias, Martin Grohe, and Daniel Marx.
\newblock Size bounds and query plans for relational joins.
\newblock {\em SIAM Journal on Computing}, 42(4):1737--1767, 2013.

\bibitem{bdg07}
Guillaume Bagan, Arnaud Durand, and Etienne Grandjean.
\newblock On acyclic conjunctive queries and constant delay enumeration.
\newblock In {\em Computer Science Logic}, pages 208--222, 2007.

\bibitem{bdes13}
Michael~J Bannister, Christopher DuBois, David Eppstein, and Padhraic Smyth.
\newblock Windows into relational events: Data structures for contiguous
  subsequences of edges.
\newblock In {\em Proceedings of the Annual {ACM-SIAM} Symposium on Discrete
  Algorithms ({SODA})}, pages 856--864, 2013.

\bibitem{b79}
Jon~Louis Bentley.
\newblock Decomposable searching problems.
\newblock {\em Information Processing Letters ({IPL})}, 8(5):244--251, 1979.

\bibitem{bps21}
Suman~K. Bera, Noujan Pashanasangi, and C.~Seshadhri.
\newblock Near-linear time homomorphism counting in bounded degeneracy graphs:
  The barrier of long induced cycles.
\newblock In {\em Proceedings of the Annual {ACM-SIAM} Symposium on Discrete
  Algorithms ({SODA})}, pages 2315--2332, 2021.

\bibitem{bks17}
Christoph Berkholz, Jens Keppeler, and Nicole Schweikardt.
\newblock Answering conjunctive queries under updates.
\newblock In {\em Proceedings of ACM Symposium on Principles of Database
  Systems ({PODS})}, pages 303--318, 2017.

\bibitem{bkk17}
Andreas Bjorklund, Petteri Kaski, and Lukasz Kowalik.
\newblock Counting thin subgraphs via packings faster than meet-in-the-middle
  time.
\newblock {\em ACM Transactions on Algorithms}, 13(4):48:1--48:26, 2017.

\bibitem{bpwz14}
Andreas Bjorklund, Rasmus Pagh, Virginia~Vassilevska Williams, and Uri Zwick.
\newblock Listing triangles.
\newblock In {\em Proceedings of International Colloquium on Automata,
  Languages and Programming ({ICALP})}, pages 223--234, 2014.

\bibitem{ck21}
Nofar Carmeli and Markus Kroll.
\newblock On the enumeration complexity of unions of conjunctive queries.
\newblock {\em ACM Transactions on Database Systems ({TODS})}, 46(2):5:1--5:41,
  2021.

\bibitem{czb+20}
Nofar Carmeli, Shai Zeevi, Christoph Berkholz, Benny Kimelfeld, and Nicole
  Schweikardt.
\newblock Answering (unions of) conjunctive queries using random access and
  random-order enumeration.
\newblock In {\em Proceedings of ACM Symposium on Principles of Database
  Systems ({PODS})}, pages 393--409, 2020.

\bibitem{cm19}
Farah Chanchary and Anil Maheshwari.
\newblock Time windowed data structures for graphs.
\newblock {\em J. Graph Algorithms Appl.}, 23(2):191--226, 2019.

\bibitem{cms20}
Farah Chanchary, Anil Maheshwari, and Michiel Smid.
\newblock Querying relational event graphs using colored range searching data
  structures.
\newblock {\em Discrete Applied Mathematics}, 286:51--61, 2020.

\bibitem{c88}
Bernard Chazelle.
\newblock A functional approach to data structures and its use in
  multidimensional searching.
\newblock {\em SIAM Journal of Computing}, 17(3):427--462, 1988.

\bibitem{cn85}
N.~Chiba and T.~Nishizeki.
\newblock Arboricity and subgraph listing algorithms.
\newblock {\em SIAM Journal of Computing}, 14(1):210--223, 1985.

\bibitem{cdm17}
Radu Curticapean, Holger Dell, and D{\'{a}}niel Marx.
\newblock Homomorphisms are a good basis for counting small subgraphs.
\newblock In {\em Proceedings of {ACM} Symposium on Theory of Computing
  ({STOC})}, pages 210--223, 2017.

\bibitem{bcko08}
Mark de~Berg, Otfried Cheong, Marc van Kreveld, and Mark Overmars.
\newblock {\em Computational Geometry: Algorithms and Applications}.
\newblock Springer-Verlag, 3rd edition, 2008.

\bibitem{dk18}
Shaleen Deep and Paraschos Koutris.
\newblock Compressed representations of conjunctive query results.
\newblock In {\em Proceedings of ACM Symposium on Principles of Database
  Systems ({PODS})}, pages 307--322, 2018.

\bibitem{d20}
Arnaud Durand.
\newblock Fine-grained complexity analysis of queries: From decision to
  counting and enumeration.
\newblock In {\em Proceedings of ACM Symposium on Principles of Database
  Systems ({PODS})}, pages 331--346, 2020.

\bibitem{dg07}
Arnaud Durand and Etienne Grandjean.
\newblock First-order queries on structures of bounded degree are computable
  with constant delay.
\newblock {\em {ACM} Trans. Comput. Log.}, 8(4):21, 2007.

\bibitem{dss14}
Arnaud Durand, Nicole Schweikardt, and Luc Segoufin.
\newblock Enumerating answers to first-order queries over databases of low
  degree.
\newblock In {\em Proceedings of ACM Symposium on Principles of Database
  Systems ({PODS})}, pages 121--131, 2014.

\bibitem{e94}
David Eppstein.
\newblock Arboricity and bipartite subgraph listing algorithms.
\newblock {\em Information Processing Letters ({IPL})}, 51(4):207--211, 1994.

\bibitem{e99}
David Eppstein.
\newblock Subgraph isomorphism in planar graphs and related problems.
\newblock {\em J. Graph Algorithms Appl.}, 3(3):1--27, 1999.

\bibitem{els10}
David Eppstein, Maarten Loffler, and Darren Strash.
\newblock Listing all maximal cliques in sparse graphs in near-optimal time.
\newblock In {\em International Symposium on Algorithms and Computation
  ({ISAAC})}, volume 6506, pages 403--414, 2010.

\bibitem{fkll15}
Peter Floderus, Miroslaw Kowaluk, Andrzej Lingas, and Eva{-}Marta Lundell.
\newblock Detecting and counting small pattern graphs.
\newblock {\em {SIAM} J. Discret. Math.}, 29(3):1322--1339, 2015.

\bibitem{flr+12}
Fedor~V. Fomin, Daniel Lokshtanov, Venkatesh Raman, Saket Saurabh, and
  B.~V.~Raghavendra Rao.
\newblock Faster algorithms for finding and counting subgraphs.
\newblock {\em Journal of Computer and System Sciences ({JCSS})},
  78(3):698--706, 2012.

\bibitem{f04}
Ehud Friedgut.
\newblock Hypergraphs, entropy, and inequalities.
\newblock {\em Am. Math. Mon.}, 111(9):749--760, 2004.

\bibitem{gkw19}
Pierre{-}Louis Giscard, Nils~M. Kriege, and Richard~C. Wilson.
\newblock A general purpose algorithm for counting simple cycles and simple
  paths of any length.
\newblock {\em Algorithmica}, 81(7):2716--2737, 2019.

\bibitem{gklp17}
Isaac Goldstein, Tsvi Kopelowitz, Moshe Lewenstein, and Ely Porat.
\newblock Conditional lower bounds for space/time tradeoffs.
\newblock In {\em Algorithms and Data Structures Workshop ({WADS})}, pages
  421--436. Springer, 2017.

\bibitem{glp19}
Isaac Goldstein, Moshe Lewenstein, and Ely Porat.
\newblock On the hardness of set disjointness and set intersection with bounded
  universe.
\newblock In {\em International Symposium on Algorithms and Computation
  ({ISAAC})}, pages 7:1--7:22, 2019.

\bibitem{hkss13}
Chinh~T. Hoang, Marcin Kaminski, Joe Sawada, and R.~Sritharan.
\newblock Finding and listing induced paths and cycles.
\newblock {\em Discrete Applied Mathematics}, 161(4-5):633--641, 2013.

\bibitem{kno+20}
Ahmet Kara, Milos Nikolic, Dan Olteanu, and Haozhe Zhang.
\newblock Trade-offs in static and dynamic evaluation of hierarchical queries.
\newblock In {\em Proceedings of ACM Symposium on Principles of Database
  Systems ({PODS})}, pages 375--392, 2020.

\bibitem{ks13}
Wojciech Kazana and Luc Segoufin.
\newblock Enumeration of first-order queries on classes of structures with
  bounded expansion.
\newblock In {\em Proceedings of ACM Symposium on Principles of Database
  Systems ({PODS})}, pages 297--308, 2013.

\bibitem{knrr15}
Mahmoud~Abo Khamis, Hung~Q. Ngo, Christopher R{\'{e}}, and Atri Rudra.
\newblock Joins via geometric resolutions: Worst-case and beyond.
\newblock In {\em Proceedings of ACM Symposium on Principles of Database
  Systems ({PODS})}, pages 213--228, 2015.

\bibitem{kkm00}
Ton Kloks, Dieter Kratsch, and Haiko M{\"{u}}ller.
\newblock Finding and counting small induced subgraphs efficiently.
\newblock {\em Information Processing Letters ({IPL})}, 74(3-4):115--121, 2000.

\bibitem{np85}
Jaroslav Nesetril and Svatopluk Poljak.
\newblock On the complexity of the subgraph problem.
\newblock {\em Commentationes Mathematicae Universitatis Carolinae},
  26(2):415--419, 1985.

\bibitem{nprr18}
Hung~Q. Ngo, Ely Porat, Christopher Re, and Atri Rudra.
\newblock Worst-case optimal join algorithms.
\newblock {\em Journal of the ACM ({JACM})}, 65(3):16:1--16:40, 2018.

\bibitem{nrr13}
Hung~Q. Ngo, Christopher Re, and Atri Rudra.
\newblock Skew strikes back: new developments in the theory of join algorithms.
\newblock {\em {SIGMOD} Rec.}, 42(4):5--16, 2013.

\bibitem{oz15}
Dan Olteanu and Jakub Zavodny.
\newblock Size bounds for factorised representations of query results.
\newblock {\em ACM Transactions on Database Systems ({TODS})}, 40(1):2:1--2:44,
  2015.

\bibitem{s15}
Saladi Rahul.
\newblock Improved bounds for orthogonal point enclosure query and point
  location in orthogonal subdivisions in $\mathbb{R}^3$.
\newblock In {\em Proceedings of the Annual {ACM-SIAM} Symposium on Discrete
  Algorithms ({SODA})}, pages 200--211, 2015.

\bibitem{ssv18}
Nicole Schweikardt, Luc Segoufin, and Alexandre Vigny.
\newblock Enumeration for {FO} queries over nowhere dense graphs.
\newblock In {\em Proceedings of ACM Symposium on Principles of Database
  Systems ({PODS})}, pages 151--163, 2018.

\bibitem{s15b}
Luc Segoufin.
\newblock Constant delay enumeration for conjunctive queries.
\newblock {\em {SIGMOD} Rec.}, 44(1):10--17, 2015.

\bibitem{sv17}
Luc Segoufin and Alexandre Vigny.
\newblock Constant delay enumeration for {FO} queries over databases with local
  bounded expansion.
\newblock In {\em Proceedings of International Conference on Database Theory
  ({ICDT})}, volume~68, pages 20:1--20:16, 2017.

\bibitem{t22}
Yufei Tao.
\newblock Algorithmic techniques for independent query sampling.
\newblock In {\em Proceedings of ACM Symposium on Principles of Database
  Systems ({PODS})}, 2022.

\bibitem{ty22}
Yufei Tao and Ke~Yi.
\newblock Intersection joins under updates.
\newblock {\em Journal of Computer and System Sciences ({JCSS})}, 124:41--64,
  2022.

\bibitem{v14}
Todd~L. Veldhuizen.
\newblock Triejoin: {A} simple, worst-case optimal join algorithm.
\newblock In {\em Proceedings of International Conference on Database Theory
  ({ICDT})}, pages 96--106, 2014.

\bibitem{ws98}
Duncan~J. Watts and Steven~H. Strogatz.
\newblock Collective dynamics of `small-world' networks.
\newblock {\em Nature}, 393:440--442, 1998.

\bibitem{ww13}
Virginia~Vassilevska Williams and Ryan Williams.
\newblock Finding, minimizing, and counting weighted subgraphs.
\newblock {\em SIAM Journal of Computing}, 42(3):831--854, 2013.

\end{thebibliography}



\end{document}